%% file: main.tex
\documentclass[10pt,journal,letterpaper]{IEEEtran}
\pdfoutput=1
\usepackage[T1]{fontenc}
\usepackage{amsmath,amssymb,amsfonts,bm}
\usepackage{booktabs,array,tabularx,ragged2e}
\usepackage{cite}
\usepackage{graphicx}
\usepackage{microtype}
\usepackage{xcolor}
\usepackage{flushend}
\usepackage{needspace}
\usepackage{etoolbox}
\usepackage{hyperref}

\patchcmd{\section}
  {\normalfont\normalsize\centering\scshape}
  {\normalfont\fontsize{11}{13}\selectfont\centering\bfseries}
  {}
  {\PackageError{arxiv-layout}
    {Could not update section heading style}{}}

\graphicspath{{figures/}}
\hypersetup{
  colorlinks=true,
  linkcolor=blue,
  citecolor=blue,
  urlcolor=black,
  pdftitle={Phase-Resolved Molecular Timing Channels under Pulsatile Drift: Volterra Characterization and Corrected Inverse-Gaussian Approximation},
  pdfauthor={Yen-Chi Lee},
  pdfsubject={Phase-resolved first-hitting molecular timing channels under deterministic pulsatile drift},
  pdfkeywords={Channel modeling, corrected inverse-Gaussian approximation, first-hitting time, molecular communication, pulsatile drift, release phase, Volterra integral equation}
}

\newcommand{\Prob}{\mathbb{P}}
\newcommand{\R}{\mathbb{R}}
\newcommand{\Torus}{\mathbb{T}}

\newtheorem{proposition}{Proposition}
\newtheorem{theorem}{Theorem}
\newtheorem{lemma}{Lemma}
\newtheorem{corollary}{Corollary}
\newcommand{\secref}[2]{Section~\ref{#1}-\ref{#2}}
\newcommand{\appsubref}[2]{Appendix~\ref{#1}-\ref{#2}}
\renewcommand{\thetable}{\arabic{table}}
\renewcommand{\thesubsection}{\Alph{subsection}}

\begin{document}

\title{Phase-Resolved Molecular Timing Channels under Pulsatile Drift: Volterra Characterization and Corrected Inverse-Gaussian Approximation}

\author{Yen-Chi Lee,~\IEEEmembership{Member,~IEEE}%
\thanks{Yen-Chi Lee is with the Department of Mathematics, National Central University, Taoyuan, Taiwan (e-mail: yclee@math.ncu.edu.tw).}%
\thanks{This work was supported by the National Science and Technology Council of Taiwan under Grant NSTC 113-2115-M-008-013-MY3.
This paper extends our conference manuscript~\cite{Lee2026CIG}.}}

\maketitle
\thispagestyle{plain}

\begin{abstract}
Pulsatile drift makes molecular arrival times depend on the
release phase. We study a one-dimensional absorbing timing
channel with constant diffusivity and positive-mean sinusoidal drift.
A moving-boundary formulation yields an exact Volterra
characterization of the phase-conditioned first-hitting law,
with regularity, uniqueness, properness, and moment properties.
Within this framework, we integrate and assess the corrected
inverse-Gaussian (C-IG) approximation introduced in our conference
work. Comparisons with a numerically validated Volterra reference
distinguish raw-kernel mass discrepancy from normalized shape
errors and show that close cumulative agreement can coexist with
density and peak-count discrepancies, particularly under
transient backflow. Phase-family analysis further shows that
 a small between-phase contribution to delay variance can coexist
with substantial differences between conditional cumulative distributions.
We connect these conditional laws and protocol-induced mixtures
to timing likelihoods and absorption probabilities within a
prescribed reception window, clarifying the arrival-time
structures preserved by C-IG and the phase information
concealed by averaging.
\end{abstract}


\input{sections/01_introduction}
\input{sections/02_system_model}
\input{sections/03_exact_characterization}
\input{sections/04_cig_approximation}
\input{sections/05_numerical_methods}
\input{sections/06_cig_accuracy}
\input{sections/07_phase_structure}
\input{sections/08_conclusion}

\appendices

\section{Proof of the Phase-resolved Volterra Characterization}
\label{app:volterra-proof}

Fix $\eta\in\Torus$ and a finite horizon $T_{\mathrm h}>0$.
For this proof, write $b=b_\eta$ and $f=f_\eta$.
The boundary in~\eqref{eq:physical-boundary} satisfies
$b\in C^\infty([0,T_{\mathrm h}])$ and
$b(0)=d/\sqrt{2D}>0$.
Classical moving-boundary theory therefore provides a continuous
first-passage density on $(0,T_{\mathrm h}]$
\cite{Peskir2002,Lehmann2002}.
At this stage, only the subdensity property
$\int_0^t f(u)\,du\leq1$ is needed.

\subsection{First-contact decomposition}

Let $T=\inf\{t>0:W(t)=b(t)\}$.
For fixed $t>0$ and $x\geq b(t)$, any continuous path with
$W(t)\geq x$ must have crossed the boundary by time $t$.
Conditioning on $T=u$ and using the strong Markov property gives
\begin{equation}
  \Prob(W(t)\geq x)
  =
  \int_0^t
  \overline\Phi\!\left(
    \frac{x-b(u)}{\sqrt{t-u}}
  \right)f(u)\,du.
  \label{eq:app-first-contact-tail}
\end{equation}
Setting $x=b(t)$ yields~\eqref{eq:fortet-tail-identity}.

For $x>b(t)$, differentiation with respect to $x$ gives
\begin{equation}
  p_t(x)
  =
  \int_0^t p_{t-u}(x-b(u))f(u)\,du.
  \label{eq:app-first-contact-density}
\end{equation}
To pass to $x\downarrow b(t)$, split the integral at $t/2$.
On $[0,t/2]$, the Gaussian kernel is uniformly bounded and
$f$ is integrable. On $[t/2,t)$, continuity of $f$ gives
a finite bound, while
$p_{t-u}(x-b(u))\leq[2\pi(t-u)]^{-1/2}$.
The latter bound is integrable in $u$.
Dominated convergence therefore yields
\eqref{eq:fortet-density-identity}.

\subsection{Second-kind equation}

Define
\begin{equation}
  H(t)=\overline\Phi\!\left(\frac{b(t)}{\sqrt t}\right).
  \label{eq:app-fortet-left-side}
\end{equation}
Direct differentiation gives
\begin{equation}
  H'(t)
  =
  \left[\frac{b(t)}{2t}-b'(t)\right]p_t(b(t)).
  \label{eq:app-fortet-left-derivative}
\end{equation}
For $u<t$, put
\begin{equation}
  \Psi(t,u)=
  \overline\Phi\!\left(
    \frac{b(t)-b(u)}{\sqrt{t-u}}
  \right).
  \label{eq:app-fortet-tail-kernel}
\end{equation}
Since $b$ is differentiable,
$\Psi(t,u)\to1/2$ as $u\uparrow t$.
Moreover,
\begin{equation}
  \begin{aligned}
    \partial_t\Psi(t,u)
    &=
    -\left[
      b'(t)-\frac{b(t)-b(u)}{2(t-u)}
    \right]\\
    &\quad\times
    p_{t-u}(b(t)-b(u)).
  \end{aligned}
  \label{eq:app-fortet-kernel-derivative}
\end{equation}
On compact time intervals away from zero, this derivative has
at most an integrable $(t-u)^{-1/2}$ singularity.
Thus, differentiation of the Fortet identity gives
\begin{equation}
  \begin{aligned}
    H'(t)
    &=\frac12 f(t)\\
    &\quad-
    \int_0^t
    \left[
      b'(t)-\frac{b(t)-b(u)}{2(t-u)}
    \right]\\
    &\hspace{35mm}\times
    p_{t-u}(b(t)-b(u))f(u)\,du.
  \end{aligned}
  \label{eq:app-differentiated-fortet}
\end{equation}
Combining this expression with
\eqref{eq:app-fortet-left-derivative} yields
\begin{equation}
  \begin{aligned}
    f(t)
    &=
    \left[\frac{b(t)}{t}-2b'(t)\right]p_t(b(t))\\
    &\quad+
    \int_0^t
    \left[
      2b'(t)-\frac{b(t)-b(u)}{t-u}
    \right]\\
    &\hspace{35mm}\times
    p_{t-u}(b(t)-b(u))f(u)\,du.
  \end{aligned}
  \label{eq:app-intermediate-volterra}
\end{equation}
Finally, the density-level Fortet identity implies
\begin{equation}
  b'(t)p_t(b(t))
  =
  \int_0^t
    b'(t)p_{t-u}(b(t)-b(u))f(u)\,du.
  \label{eq:app-density-cancellation}
\end{equation}
Using this equality in
\eqref{eq:app-intermediate-volterra} gives
\eqref{eq:volterra-second-kind} with the forcing and kernel in
\eqref{eq:volterra-forcing} and~\eqref{eq:volterra-kernel}.
This is the classical smooth-boundary second-kind formulation
\cite{Durbin1985,Buonocore1987,Peskir2002} specialized to
$b_\eta$.

\subsection{Continuous extensions and uniqueness}

Let $b_*=b(0)>0$. For sufficiently small $t>0$,
$b(t)\geq b_*/2$, while $b$ and $b'$ remain bounded.
Consequently, for constants $C_1,C_2>0$,
\begin{equation}
  |h_\eta(t)|
  \leq
  C_1\left(t^{-3/2}+t^{-1/2}\right)
  \exp\!\left(-\frac{C_2}{t}\right)
  \longrightarrow0.
  \label{eq:app-forcing-origin}
\end{equation}
Thus $h_\eta(0)=0$ is a continuous extension.

Appendix~\ref{app:kernel-bound-proof} establishes a uniform
bound
\begin{equation}
  |K_\eta(t,u)|
  \leq C_3\sqrt{t-u},
  \qquad 0\leq u<t\leq T_{\mathrm h},
  \label{eq:app-kernel-continuity}
\end{equation}
for a finite constant $C_3$.
The kernel is continuous off the diagonal, and this bound
extends it continuously to the closed triangle by setting
$K_\eta(t,t)=0$.

Using the subdensity property in the Volterra equation gives
\begin{equation}
  \begin{aligned}
    |f_\eta(t)|
    &\leq |h_\eta(t)|
      +C_3\sqrt t\int_0^t f_\eta(u)\,du\\
    &\leq |h_\eta(t)|+C_3\sqrt t
      \longrightarrow0.
  \end{aligned}
  \label{eq:app-density-origin}
\end{equation}
Hence $f_\eta(0)=0$ also gives a continuous extension.

Suppose $g_1,g_2\in C([0,T_{\mathrm h}])$ solve
\eqref{eq:volterra-second-kind}, and put $z=g_1-g_2$.
Continuity of the kernel on the compact triangle implies that
$C_K=\max_{\Delta_{T_{\mathrm h}}}|K_\eta|<\infty$.
Therefore
\begin{equation}
  |z(t)|\leq C_K\int_0^t|z(u)|\,du.
  \label{eq:app-volterra-uniqueness}
\end{equation}
Gronwall's inequality gives $z\equiv0$.
The first-passage density supplies existence, and the preceding
argument establishes uniqueness in $C([0,T_{\mathrm h}])$.
This proves Theorem~\ref{thm:phase-resolved-volterra}.

\section{Kernel Regularity and Dimensionless Bounds}
\label{app:kernel-bound-proof}

\subsection{Secant--tangent cancellation}

Let $b\in C^2([0,T_{\mathrm h}])$ and
$L_b=\sup_{0\leq r\leq T_{\mathrm h}}|b''(r)|$.
For $0\leq u<t\leq T_{\mathrm h}$, the fundamental theorem
of calculus gives
\begin{equation}
  \begin{aligned}
    b'(t)-\frac{b(t)-b(u)}{t-u}
    &=
    \frac{1}{t-u}\int_u^t[b'(t)-b'(r)]\,dr\\
    &=
    \frac{1}{t-u}\int_u^t\int_r^t b''(z)\,dz\,dr.
  \end{aligned}
  \label{eq:app-secant-tangent-representation}
\end{equation}
It follows that
\begin{equation}
  \left|
    b'(t)-\frac{b(t)-b(u)}{t-u}
  \right|
  \leq
  \frac{L_b}{t-u}\int_u^t(t-r)\,dr
  =\frac{L_b}{2}(t-u).
  \label{eq:app-secant-tangent-bound}
\end{equation}
Since
$p_{t-u}(x)\leq[2\pi(t-u)]^{-1/2}$,
multiplication yields
\begin{equation}
  |K_b(t,u)|
  \leq
  \frac{L_b}{2\sqrt{2\pi}}\sqrt{t-u}.
  \label{eq:app-general-kernel-bound}
\end{equation}
This proves the general bound in
Lemma~\ref{lem:curvature-kernel-bound}.

For the moving boundary $b_\eta$,
\eqref{eq:boundary-derivatives} implies
\begin{equation}
  \sup_{0\leq r\leq T_{\mathrm h}}
  |b_\eta''(r)|
  \leq\frac{a\omega}{\sqrt{2D}}.
  \label{eq:app-physical-curvature-bound}
\end{equation}
Substitution into~\eqref{eq:app-general-kernel-bound} gives
\eqref{eq:physical-kernel-curvature-bound}, uniformly in
$\eta$.

\subsection{Brownian scaling}

Recall that $t_a=d/v_0$ and
$B(\tau)=W(t_a\tau)/\sqrt{t_a}$.
The moving boundary $b_\eta$ before nondimensionalization
therefore becomes
\begin{equation}
  \begin{aligned}
    \beta_\eta(\tau)
    &=\frac{b_\eta(t_a\tau)}{\sqrt{t_a}}\\
    &=\frac{d}{\sqrt{2Dt_a}}
      [1-m_\eta(\tau)]\\
    &=\sqrt{\frac{\mathrm{Pe}}{2}}\,
      [1-m_\eta(\tau)].
  \end{aligned}
  \label{eq:app-boundary-scaling}
\end{equation}
The Brownian transition density obeys
\begin{equation}
  p_{t_a r}(\sqrt{t_a}\,x)
  =t_a^{-1/2}p_r(x).
  \label{eq:app-gaussian-scaling}
\end{equation}
Together with
$b_\eta'(t_a\tau)=t_a^{-1/2}\beta_\eta'(\tau)$,
this gives
\begin{equation}
  \begin{aligned}
    t_a h_\eta(t_a\tau)
    &=
    \left[
      \frac{\beta_\eta(\tau)}{\tau}
      -\beta_\eta'(\tau)
    \right]p_\tau(\beta_\eta(\tau)),\\
    t_a K_\eta(t_a\tau,t_a\xi)
    &=
    \left[
      \beta_\eta'(\tau)
      -\frac{\beta_\eta(\tau)-\beta_\eta(\xi)}
             {\tau-\xi}
    \right]\\
    &\quad\times
    p_{\tau-\xi}(\beta_\eta(\tau)-\beta_\eta(\xi)).
  \end{aligned}
  \label{eq:app-scaled-volterra-coefficients}
\end{equation}
These are precisely the forcing and kernel obtained by applying
the standard-Brownian construction directly to $\beta_\eta$.

Alternatively, multiplying the physical Volterra equation at
$t=t_a\tau$ by $t_a$, changing variables $u=t_a\xi$, and using
$\widetilde f_\eta(\xi)=t_a f_\eta(t_a\xi)$ gives
\eqref{eq:dimensionless-volterra-equation} and
\eqref{eq:dimensionless-volterra-coefficients}.

\subsection{Dimensionless operator bound}

The dimensionless boundary satisfies
\begin{equation}
  \sup_{0\leq r\leq\tau_{\mathrm h}}
  |\beta_\eta''(r)|
  \leq
  \sqrt{\frac{\mathrm{Pe}}{2}}\,\alpha\Omega.
  \label{eq:app-dimensionless-curvature-bound}
\end{equation}
Applying~\eqref{eq:app-general-kernel-bound} to $\beta_\eta$
therefore gives
\begin{equation}
  |\widetilde K_\eta(\tau,\xi)|
  \leq
  \frac{\alpha\Omega\sqrt{\mathrm{Pe}}}{4\sqrt{\pi}}
  \sqrt{\tau-\xi}.
  \label{eq:app-dimensionless-kernel-bound}
\end{equation}
For $g\in C([0,\tau_{\mathrm h}])$,
\begin{equation}
  \begin{aligned}
    |(\widetilde{\mathcal K}_\eta g)(\tau)|
    &\leq
    \frac{\alpha\Omega\sqrt{\mathrm{Pe}}}{4\sqrt{\pi}}
    \|g\|_\infty
    \int_0^\tau\sqrt{\tau-\xi}\,d\xi\\
    &=
    \frac{\alpha\Omega\sqrt{\mathrm{Pe}}}{6\sqrt{\pi}}
    \tau^{3/2}\|g\|_\infty.
  \end{aligned}
  \label{eq:app-volterra-operator-estimate}
\end{equation}
Taking the supremum over $0\leq\tau\leq\tau_{\mathrm h}$
proves~\eqref{eq:dimensionless-operator-bound}.

\section{Properness, Exponential Moments, and Stopped-moment Identities}
\label{app:properness-proof}

\subsection{Pathwise comparison}

For the sinusoidal drift,
\begin{equation}
  R_\eta(t)
  =\frac{a}{\omega}
    [\cos\eta-\cos(\eta+\omega t)]
  \geq-\frac{a}{\omega}(1-\cos\eta)
  =-C_\eta.
  \label{eq:app-displacement-lower-bound}
\end{equation}
Let $L_\eta=d+C_\eta$ and define
\begin{equation}
  \begin{aligned}
    U(t)&=v_0t+\sqrt{2D}\,W(t),\\
    T_\eta^+&=\inf\{t>0:U(t)=L_\eta\}.
  \end{aligned}
  \label{eq:app-comparison-process}
\end{equation}
Since $v_0>0$, $T_\eta^+$ is finite almost surely.
At this time,
\begin{equation}
  X_\eta(T_\eta^+)
  =U(T_\eta^+)+R_\eta(T_\eta^+)
  \geq L_\eta-C_\eta=d.
  \label{eq:app-comparison-at-hitting}
\end{equation}
Continuity and $X_\eta(0)=0<d$ imply
\begin{equation}
  T_\eta\leq T_\eta^+
  \qquad\text{almost surely}.
  \label{eq:app-pathwise-domination}
\end{equation}
Thus $T_\eta$ is also finite almost surely.

\subsection{Exponential integrability}

The comparison time has the positive-drift IG density
\cite{Srinivas2012}
\begin{equation}
  g_{L_\eta}(t)
  =
  \frac{L_\eta}{\sqrt{4\pi Dt^3}}
  \exp\!\left[
    -\frac{(L_\eta-v_0t)^2}{4Dt}
  \right],
  \qquad t>0.
  \label{eq:app-comparison-ig-density}
\end{equation}
For $0\leq\theta<v_0^2/(4D)$, expansion of the exponent gives
\begin{equation}
  \begin{aligned}
    \mathbb E[e^{\theta T_\eta^+}]
    &=
    \frac{L_\eta e^{L_\eta v_0/(2D)}}{\sqrt{4\pi D}}
    \int_0^\infty t^{-3/2}\\
    &\quad\times
    \exp\!\left[
      -\frac{L_\eta^2}{4Dt}
      -\left(\frac{v_0^2}{4D}-\theta\right)t
    \right]dt\\
    &=
    \exp\!\left\{
      \frac{L_\eta}{2D}
      \left[v_0-\sqrt{v_0^2-4D\theta}\right]
    \right\}.
  \end{aligned}
  \label{eq:app-comparison-exponential-moment}
\end{equation}
At $\theta=v_0^2/(4D)$, the large-time integrand is proportional
to $t^{-3/2}$ and remains integrable. Equivalently, monotone
convergence extends the last expression to the endpoint.

Since $\theta\geq0$, pathwise domination gives
\begin{equation}
  \mathbb E[e^{\theta T_\eta}]
  \leq\mathbb E[e^{\theta T_\eta^+}],
  \label{eq:app-exponential-moment-domination}
\end{equation}
which proves~\eqref{eq:exponential-moment-bound}.

For any fixed $\theta_0\in(0,v_0^2/(4D)]$ and integer $n\geq1$,
the inequality $t^n\leq n!\theta_0^{-n}e^{\theta_0t}$ gives
\begin{equation}
  \mathbb E[T_\eta^n]
  \leq
  \frac{n!}{\theta_0^n}
  \exp\!\left\{
    \frac{d+C_\eta}{2D}
    \left[v_0-\sqrt{v_0^2-4D\theta_0}\right]
  \right\}.
  \label{eq:app-polynomial-moment-bound}
\end{equation}
The upper bound $C_\eta\leq2a/\omega$ makes these estimates
uniform over release phase. This completes the proof of
Proposition~\ref{prop:properness}.

For a continuous periodic drift $v$ with period $P$ and
positive cycle mean $v_0$, the centered displacement
\begin{equation}
  R_s(t)=\int_0^t[v(s+r)-v_0]\,dr
  \label{eq:app-general-periodic-remainder}
\end{equation}
satisfies $R_s(t+P)=R_s(t)$.
It is therefore bounded below, and the same comparison
argument applies.

\subsection{Justification of the stopped-moment identities}

Fix $\eta$ and write $T=T_\eta$.
Set $T_n=T\wedge n$.
Optional stopping at these bounded times gives
\begin{equation}
  \mathbb E[W(T_n)]=0,
  \qquad
  \mathbb E[W(T_n)^2]=\mathbb E[T_n].
  \label{eq:app-bounded-stopping-identities}
\end{equation}
For integers $m\geq n$, the It\^o isometry also gives
\begin{equation}
  \mathbb E\!\left[
    |W(T_m)-W(T_n)|^2
  \right]
  =\mathbb E[T_m-T_n].
  \label{eq:app-stopped-brownian-l2}
\end{equation}
Because $\mathbb E[T]<\infty$, the right-hand side tends to
zero as $m,n\to\infty$. Thus $W(T_n)$ is Cauchy in $L^2$.
Almost-sure finiteness of $T$ and continuity of Brownian paths
identify its limit as $W(T)$.

Passing to the limit in
\eqref{eq:app-bounded-stopping-identities} yields
\begin{equation}
  \mathbb E[W(T)]=0,
  \qquad
  \mathbb E[W(T)^2]=\mathbb E[T].
  \label{eq:app-unbounded-stopping-identities}
\end{equation}
At first absorption, path continuity gives
\begin{equation}
  d-M_\eta(T_\eta)=\sqrt{2D}\,W(T_\eta).
  \label{eq:app-stopped-hitting-condition}
\end{equation}
Taking expectations and second moments therefore gives
\begin{equation}
  \begin{aligned}
    \mathbb E[d-M_\eta(T_\eta)]&=0,\\
    \mathbb E[(d-M_\eta(T_\eta))^2]
      &=2D\,\mathbb E[T_\eta].
  \end{aligned}
  \label{eq:app-stopped-displacement-moments}
\end{equation}
Since $M_\eta(t)=v_0t+R_\eta(t)$ with bounded $R_\eta$,
the polynomial moment bounds above also establish absolute
integrability of the corresponding expressions.
Writing the expectations against $f_\eta$ proves
\eqref{eq:first-stopped-moment-identity} and
\eqref{eq:second-stopped-moment-identity}.

\section{Affine Boundaries and Inverse-Gaussian Recovery}
\label{app:stationary-ig-proof}

Let $b\in C^1([0,T_{\mathrm h}])$ with $T_{\mathrm h}>0$.
If $b$ is affine, then
\begin{equation}
  b'(t)=\frac{b(t)-b(u)}{t-u},
  \qquad 0\leq u<t\leq T_{\mathrm h},
  \label{eq:app-affine-secant-slope}
\end{equation}
so~\eqref{eq:volterra-kernel} gives $K_b(t,u)=0$.

Conversely, suppose $K_b(t,u)=0$ for every $u<t$.
The Gaussian factor in the kernel is strictly positive, hence
\eqref{eq:app-affine-secant-slope} holds.
Taking $t=T_{\mathrm h}$ gives
\begin{equation}
  b(u)
  =
  b(T_{\mathrm h})
  -b'(T_{\mathrm h})(T_{\mathrm h}-u),
  \qquad 0\leq u<T_{\mathrm h}.
  \label{eq:app-affine-boundary-reconstruction}
\end{equation}
Continuity extends this expression to $u=T_{\mathrm h}$,
so $b$ is affine on the entire interval.

For the sinusoidal family,
\begin{equation}
  b_\eta''(t)
  =
  -\frac{a\omega}{\sqrt{2D}}
  \cos(\eta+\omega t).
  \label{eq:app-sinusoidal-affinity}
\end{equation}
Since $\omega>0$, the cosine cannot vanish identically on
an interval of positive length. Thus $b_\eta$ is affine
if and only if $a=0$.

When $a=0$, the boundary is independent of release phase:
\begin{equation}
  b_\eta(t)=\frac{d-v_0t}{\sqrt{2D}},
  \qquad
  b_\eta'(t)=-\frac{v_0}{\sqrt{2D}}.
  \label{eq:app-constant-drift-boundary}
\end{equation}
Consequently,
\begin{equation}
  \frac{b_\eta(t)}{t}-b_\eta'(t)
  =\frac{d}{\sqrt{2D}\,t}.
  \label{eq:app-constant-drift-forcing-factor}
\end{equation}
Substitution into~\eqref{eq:volterra-forcing} yields
\begin{equation}
  h_\eta(t)
  =
  \frac{d}{\sqrt{4\pi Dt^3}}
  \exp\!\left[-\frac{(d-v_0t)^2}{4Dt}\right].
  \label{eq:app-recovered-ig-density}
\end{equation}
Since $K_\eta=0$, the Volterra equation gives
$f_\eta=h_\eta=f_{\mathrm{IG}}$.
This proves Corollary~\ref{cor:stationary-ig-recovery}.

\input{appendices/05_numerical_implementation}

\section*{Acknowledgment}
The author used ChatGPT (OpenAI; GPT-6 Astra) as an assistive tool for language refinement, manuscript organization, and routine debugging of selected computational scripts. All mathematical derivations, numerical methodologies, reported results, references, scientific interpretations, and conclusions were independently verified and approved by the author.

\input{references}
\end{document}

%% file: sections/01_introduction.tex
\section{Introduction}
\label{sec:introduction}

\IEEEPARstart{M}{olecular} communication (MC) uses molecular transport
to convey information, with potential applications in targeted
delivery and biological sensing~\cite{Farsad2016,Jamali2019,Chahibi2013}.
In vascular environments, cardiac pulsation motivates channel models
that account for time-varying flow and molecular release
timing~\cite{Wille2025,Symeonidis2026}.
For an absorbing receiver, the relevant observable is the
first-hitting time (FHT): the delay until a molecule first reaches
the absorbing boundary.
Different release phases expose molecules to different sequences
of forward flow and possible backflow, producing phase-dependent
arrival distributions.
Figure~\ref{fig:phase-resolved-channel-concept} illustrates the
application context and the idealized absorbing link studied here.

Under constant positive drift, the one-dimensional FHT follows
an inverse-Gaussian (IG) law, which underlies the additive IG
timing channel~\cite{Srinivas2012}.
Under pulsatile drift, each release phase selects a subsequent
drift history and a conditional arrival law.
First absorption constrains this law: molecules absorbed earlier
cannot contribute to later arrivals.
The channel therefore requires a family of phase-conditioned
first-contact distributions.
Related MC studies address mobility-induced
fluctuations~\cite{Ahmadzadeh2017Mobile}, cyclostationary
signaling~\cite{KoikeAkino2017}, time-dependent drift and
diffusion~\cite{Chouhan2022}, and passive reception under
oscillatory electric fields~\cite{Cho2022}.
For vascular transport, analytical concentration models describe
pulsatile straight and closed-loop
channels~\cite{Symeonidis2026}, while physically parameterized IG
mixtures describe transport through complex
networks~\cite{Jakumeit2025MIGHT}.
Our focus is the release-phase dependence of arrival times
subject to first absorption.

The analytical difficulty lies in enforcing first absorption under
time-dependent drift.
Molini \emph{et al.} identified a drift--diffusivity proportionality
condition for an exact method-of-images construction and examined
its approximate use under periodic drift with constant
diffusion~\cite{Molini2011}.
Transforming cumulative drift into a moving boundary instead yields
classical first-passage integral
representations~\cite{Fortet1943,Durbin1985,Buonocore1987,Peskir2002}.
These provide an exact Volterra characterization, whose numerical
evaluation supplies the reference law for this study.
Explicit approximations offer a complementary route to repeated
channel evaluation.
Using a Girsanov transformation of joint first-hitting
time--location statistics, Lee \emph{et al.} obtained an exact
three-dimensional channel impulse response under constant uniform
drift~\cite{Lee2026Exact3D}.
Building on that result, Chou \emph{et al.} approximated absorption
under time-varying electric fields through a time-averaged
effective-drift substitution~\cite{Chou2026FieldAssisted}.
Our conference work introduced the corrected inverse-Gaussian
(C-IG) approximation~\cite{Lee2026CIG}, combining a cumulative-drift
exponential core, an instantaneous-drift-dependent Gaussian
positive-part prefactor, and phase-dependent normalization.

This article integrates and extends the conference
study~\cite{Lee2026CIG}.
We retain its C-IG construction, direct Volterra comparison, and
amplitude--frequency--phase assessment.
The journal extension develops the timing-channel formulation,
establishes properties of the exact law, and examines
phase-conditioned families and protocol-induced mixtures.
The central questions are how much arrival-time structure C-IG
preserves across release phases and which conditional channel
features are concealed by phase averaging.

The contributions of this integrated study are as follows.
\begin{enumerate}
  \item \textbf{Exact phase-resolved channel characterization.}
  We formulate a one-dimensional absorbing timing channel with
  constant diffusivity and positive-mean sinusoidal drift.
  Classical first-passage identities yield a second-kind Volterra
  characterization, for which we establish regularity and
  uniqueness, together with properness and moment properties
  of the first-hitting law.
  We derive its dimensionless representation, accommodate
  transient backflow, and recover the constant-drift IG law.

\item \textbf{C-IG approximation and arrival-time structure.}
  To the best of our knowledge, our
  conference work~\cite{Lee2026CIG} introduced the first tractable
  analytical approximation in MC specifically for
  release-phase-conditioned first-hitting distributions under
  sinusoidal drift with constant diffusivity, including transient
  flow reversal.
  C-IG combines a cumulative-drift exponential core, an
  instantaneous-drift-dependent positive-part prefactor, and
  phase-dependent normalization.
  Using a numerically validated Volterra reference, we distinguish
  raw-kernel mass discrepancy from normalized density and CDF
  errors, and assess accuracy and peak-count agreement across
  sampled operating regimes.

  \item \textbf{Phase structure and communication implications.}
  Phase--time densities, variance decomposition, and pairwise
  cumulative distribution function (CDF) separation characterize
  release-phase sensitivity.
  A small between-phase variance contribution can coexist with
  substantial CDF differences.
  We identify the conditional laws entering timing likelihoods
  and the mixtures induced by release protocols, and relate
  CDF approximation errors to absorption probabilities within a prescribed reception window.
\end{enumerate}

Section~\ref{sec:phase-resolved-channel} introduces the channel model.
Sections~\ref{sec:volterra-characterization}
and~\ref{sec:cig-approximation} present the exact characterization
and C-IG approximation.
Section~\ref{sec:numerical-methods} describes numerical evaluation
and validation.
Sections~\ref{sec:cig-accuracy} and~\ref{sec:phase-structure}
examine C-IG accuracy and phase-family structure, respectively.
Section~\ref{sec:conclusion} concludes the paper.

%% file: sections/02_system_model.tex
\section{System Model and Phase-resolved Timing Channel}
\label{sec:phase-resolved-channel}
\label{sec:system-model}

We consider the one-dimensional absorbing link in
Fig.~\ref{fig:phase-resolved-channel-concept}(b), motivated by molecular
signaling and delivery in pulsatile environments as illustrated in
Fig.~\ref{fig:phase-resolved-channel-concept}(a). The model isolates how a
prescribed time-varying drift and the molecular release epoch determine the
propagation delay. It extends the constant-drift setting of the additive
inverse-Gaussian (IG) timing channel~\cite{Srinivas2012} while retaining a
point transmitter, constant diffusivity, and first absorption as the receiver
observable.

The principal notation is summarized in Table~\ref{tab:notation}.
We write $\mathbb P$, $\mathbb E$, and $\operatorname{Var}$ for
probability, expectation, and variance, respectively;
conditioning is indicated by a vertical bar.
The indicator $\mathbf 1_A(x)$ equals one when $x\in A$
and zero otherwise, and $\mathbf 1_{\{E\}}$ similarly indicates
an event $E$.
The notation $\overset{\mathrm d}{=}$ denotes equality in
distribution, and $\delta_x$ denotes a unit point mass at $x$.
We use $(z)_+=\max\{z,0\}$ and $r\wedge s=\min\{r,s\}$.
The notation $a=O(b)$ means that $|a|\leq C|b|$ for some
constant $C$ in the stated limit or scaling regime.

\input{tables/01_notation}

\begin{figure*}[!tp]
  \centering
  \includegraphics[width=0.98\textwidth]{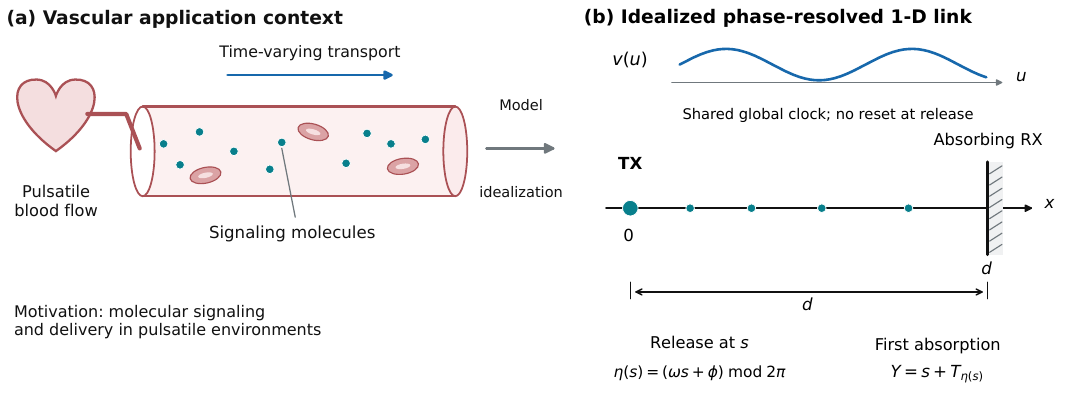}
  \caption{Application context and idealized phase-resolved molecular timing
  channel, adapted from~\cite{Lee2026CIG}. (a) Molecular signaling and
  delivery in pulsatile vascular environments motivate the prescribed
  time-varying transport. (b) The reduced model has a transmitter at $x=0$
  and a perfectly absorbing receiver at $x=d$. The spatially uniform drift
  $v(u)$ follows a shared global clock and is not reset at release. A release
  at time $s$ has phase $\eta(s)=(\omega s+\phi)\bmod 2\pi$ and absolute
  arrival time $Y=s+T_{\eta(s)}$. The waveform and particle positions are
  schematic.}
  \label{fig:phase-resolved-channel-concept}
\end{figure*}

\subsection{Global-clock drift and absorbing link}
\label{subsec:global-clock-model}

The transmitter is located at $x=0$, and the receiver is a perfectly
absorbing boundary at $x=d>0$. Before reception, a molecule moves on
$(-\infty,d)$ with constant diffusivity $D>0$ and spatially uniform drift
\begin{equation}
  v(u)=v_0+a\sin(\omega u+\phi),
  \qquad u\geq0,
  \label{eq:global-drift}
\end{equation}
where $u$ is absolute time, $v_0>0$ is the cycle-mean velocity,
$a\geq0$ is the pulsation amplitude, $\omega>0$ is the angular frequency,
and $\phi$ is the global phase offset.
The parameters $(d,D,v_0,a,\omega,\phi)$ are known, molecules are
noninteracting, and the waveform is prescribed independently
of molecular release.

For a molecule released at absolute time $s\geq0$, let $t\geq0$
denote elapsed propagation time, with $u=s+t$.
The drift--diffusion trajectory satisfies
\begin{equation}
  dX_s(t)=v(s+t)\,dt+\sqrt{2D}\,dW(t),
  \qquad X_s(0)=0,
  \label{eq:released-particle-sde}
\end{equation}
where $W$ is a standard Brownian motion.
The propagation delay is the first-hitting time
\begin{equation}
  T_s=\inf\{t>0:X_s(t)=d\},
  \label{eq:first-hitting-time}
\end{equation}
with $\inf\varnothing=\infty$.
The receiver registers and removes the molecule at this first contact.

In the vascular interpretation of
Fig.~\ref{fig:phase-resolved-channel-concept}, $x$ denotes longitudinal
position along an idealized unbranched transport segment.
For a fixed molecular species under approximately stable medium
conditions, molecular diffusivity can remain effectively constant
while cardiac pulsation varies the advective velocity.
Symeonidis \emph{et al.} use constant molecular diffusivity with
a pulsatile velocity field in their underlying three-dimensional
vascular MC model~\cite[Eq.~(1)]{Symeonidis2026}.

The present one-dimensional model isolates the effects of pulsatile
drift and release timing on first absorption.
Its spatially uniform drift excludes cross-sectional shear and
the associated longitudinal dispersion; $D$ represents molecular
diffusion.
It therefore serves as an idealized vascular reference model.
Application to a specific vessel requires justifying the spatial
reduction and the perfectly absorbing receiver assumption.

\subsection{Release phase and the conditional first-hitting law}
\label{subsec:conditional-fht-law}

The release epoch selects the phase of the global waveform. On the phase
circle $\Torus=\R/(2\pi\mathbb Z)$, define
\begin{equation}
  \eta(s)=(\omega s+\phi)\bmod 2\pi.
  \label{eq:physical-release-phase}
\end{equation}
For a fixed phase $\eta\in\Torus$, the elapsed-time drift is
\begin{equation}
  v_\eta(t)=v_0+a\sin(\eta+\omega t),
  \label{eq:phase-drift}
\end{equation}
so that $v(s+t)=v_{\eta(s)}(t)$. Its cumulative deterministic displacement is
\begin{equation}
  \begin{aligned}
    M_\eta(t)
    &=\int_0^t v_\eta(r)\,dr\\
    &=v_0t+\frac{a}{\omega}
       \left[\cos\eta-\cos(\eta+\omega t)\right].
  \end{aligned}
  \label{eq:physical-accumulated-displacement}
\end{equation}
Thus the phase-indexed trajectory and its first hitting time can be written
as
\begin{align}
  X_\eta(t)&=M_\eta(t)+\sqrt{2D}\,W(t),
  \label{eq:integrated-particle-process}\\
  T_\eta&=\inf\{t>0:X_\eta(t)=d\}.
  \label{eq:fixed-phase-hitting-time}
\end{align}
For a release at $s$, this construction gives $X_s=X_{\eta(s)}$ and
$T_s=T_{\eta(s)}$ when driven by the same Brownian path.

Let $\mu_\eta$ denote the probability law of $T_\eta$ on $(0,\infty]$, and
write its cumulative distribution function (CDF) as
\begin{equation}
  F_\eta(t)=\mu_\eta((0,t])=\Prob(T_\eta\leq t).
  \label{eq:phase-conditioned-law}
\end{equation}
The family $\mathcal C_{\mathrm{FHT}}=\{\mu_\eta:\eta\in\Torus\}$
provides the propagation component of the phase-resolved timing channel.
For the positive-mean periodic drift in~\eqref{eq:global-drift}, each law has
a continuous density $f_\eta$ on $(0,\infty)$ and unit total mass; these
properties are established in Section~\ref{sec:volterra-characterization}.
We extend $f_\eta$ by zero for $t\leq0$.

Let $P=2\pi/\omega$ be the pulsation period. A shift of the release epoch by
an integer number of periods leaves the subsequent drift history unchanged.
Coupling the trajectories with the same Brownian path therefore gives
\begin{equation}
  T_{s+kP}\overset{\mathrm d}{=}T_s,
  \qquad s+kP\geq0,
  \quad k\in\mathbb Z.
  \label{eq:periodic-law}
\end{equation}
It is consequently sufficient to characterize the conditional laws over one
cycle of release phases.

\subsection{Communication transition law}
\label{subsec:communication-transition-law}

Let $S\geq0$ be the release time selected by the communication system. We
assume that $S$ is independent of the post-release Brownian fluctuation;
conditional on $S=s$, the noise in~\eqref{eq:released-particle-sde} therefore
remains a standard Brownian motion with the same law for every $s$. The
absolute arrival time is
\begin{equation}
  Y=S+T_{\eta(S)}.
  \label{eq:phase-indexed-channel}
\end{equation}
For a Borel set $A\subseteq\R$, the communication transition kernel is
\begin{equation}
  \begin{aligned}
    Q(s,A)
    &:=\Prob(Y\in A\mid S=s)\\
    &=\int_{(0,\infty)}
        \mathbf 1_A(s+t)\,\mu_{\eta(s)}(dt),
  \end{aligned}
  \label{eq:communication-kernel-measure}
\end{equation}
and its conditional density is
\begin{equation}
  p_{Y\mid S}(y\mid s)
  =f_{\eta(s)}(y-s)\,\mathbf 1_{\{y>s\}}.
  \label{eq:communication-kernel}
\end{equation}
This expression separates the deterministic shift by $s$ from the
phase-conditioned propagation law. For likelihood-based reception, each
candidate release time selects both a delay argument $y-s$ and a physical
phase $\eta(s)$.

For $A+kP=\{y+kP:y\in A\}$, periodicity implies
\begin{equation}
  Q(s+kP,A+kP)=Q(s,A),
  \qquad s+kP\geq0.
  \label{eq:periodic-transition-kernel}
\end{equation}
When $a=0$, all phase-conditioned laws coincide with the classical
constant-drift IG law~\cite{Srinivas2012}. Under pulsation, the same additive
timing relation remains valid, but the propagation delay generally depends
on the input release time through $\eta(S)$.

\subsection{Release protocols and phase mixtures}
\label{subsec:phase-mixture}

A release protocol induces a probability measure $q$ on $\Torus$, defined
by $q(B)=\Prob(\eta(S)\in B)$ for every Borel set $B\subseteq\Torus$.
The corresponding marginal propagation-delay law is
\begin{equation}
  \mu_q(A)=\int_{\Torus}\mu_\eta(A)\,q(d\eta),
  \label{eq:phase-mixture-measure}
\end{equation}
for Borel $A\subseteq(0,\infty)$. Its CDF and density are
\begin{align}
  F_q(t)&=\int_{\Torus}F_\eta(t)\,q(d\eta),
  \label{eq:general-phase-mixture-cdf}\\
  f_q(t)&=\int_{\Torus}f_\eta(t)\,q(d\eta).
  \label{eq:general-phase-mixture-density}
\end{align}
These formulas cover both discrete phase slots and continuous phase
randomization. For example, $q=\sum_jp_j\delta_{\eta_j}$ gives
$f_q=\sum_jp_jf_{\eta_j}$, where $p_j$ is the probability of selecting phase
$\eta_j$. Uniformly randomized release phase gives
\begin{equation}
  \begin{aligned}
    \overline F(t)&=\frac{1}{2\pi}\int_0^{2\pi}F_\eta(t)\,d\eta,\\
    \overline f(t)&=\frac{1}{2\pi}\int_0^{2\pi}f_\eta(t)\,d\eta.
  \end{aligned}
  \label{eq:uniform-phase-mixture}
\end{equation}

The mixture describes the marginal delay under the specified release
protocol. Phase-aware operation instead uses the conditional law selected
by the release epoch in~\eqref{eq:communication-kernel}. The absolute-arrival
law also retains the release-time shift: it is obtained by averaging
$Q(s,A)$ over the full distribution of $S$.

\subsection{Dimensionless representation}
\label{subsec:dimensionless-model}

The nominal advection time and the dimensionless elapsed time are
\begin{equation}
  t_a=\frac{d}{v_0},
  \qquad \tau=\frac{t}{t_a}.
  \label{eq:nominal-advection-time}
\end{equation}
Introduce the dimensionless groups
\begin{equation}
  \alpha=\frac{a}{v_0},
  \qquad \Omega=\omega t_a,
  \qquad \mathrm{Pe}=\frac{v_0d}{D}.
  \label{eq:dimensionless-groups}
\end{equation}
Here, $\alpha$ is the pulsation-to-mean velocity ratio, $\Omega$ is the
waveform phase advance over one nominal transit time, and $\mathrm{Pe}$ is
the P\'eclet number measuring the strength of advection relative to diffusion
under the convention in~\eqref{eq:dimensionless-groups}; related MC models
also use this ratio to organize flow--diffusion regimes
\cite{Wicke2018DuctFlow}.

Define the scaled position and Brownian motion by
\begin{equation}
  \widetilde X_\eta(\tau)=\frac{X_\eta(t_a\tau)}{d},
  \qquad B(\tau)=\frac{W(t_a\tau)}{\sqrt{t_a}}.
  \label{eq:dimensionless-position}
\end{equation}
Brownian scaling makes $B$ a standard Brownian motion, and
\eqref{eq:released-particle-sde} becomes
\begin{equation}
  \begin{aligned}
    d\widetilde X_\eta(\tau)
    &=\left[1+\alpha\sin(\eta+\Omega\tau)\right]d\tau\\
    &\quad+\sqrt{\frac{2}{\mathrm{Pe}}}\,dB(\tau),
    \qquad \widetilde X_\eta(0)=0.
  \end{aligned}
  \label{eq:dimensionless-sde}
\end{equation}
The corresponding first hitting time and cumulative displacement are
\begin{equation}
  \mathcal T_\eta
  =\inf\{\tau>0:\widetilde X_\eta(\tau)=1\}
  =\frac{T_\eta}{t_a},
  \label{eq:dimensionless-fht}
\end{equation}
\begin{equation}
  \begin{aligned}
    m_\eta(\tau)
    &=\frac{M_\eta(t_a\tau)}{d}\\
    &=\tau+\frac{\alpha}{\Omega}
       \left[\cos\eta-\cos(\eta+\Omega\tau)\right].
  \end{aligned}
  \label{eq:dimensionless-displacement}
\end{equation}

The law of $\mathcal T_\eta$ therefore depends only on
$(\alpha,\Omega,\mathrm{Pe},\eta)$. If two physical systems have the same
values of these four quantities, their scaled trajectories satisfy the same
stochastic equation and have the same first-hitting law. In particular, for
systems $j=1,2$ with nominal advection times $t_{a,j}$,
\begin{equation}
  \frac{T_\eta^{(1)}}{t_{a,1}}
  \overset{\mathrm d}{=}
  \frac{T_\eta^{(2)}}{t_{a,2}}.
  \label{eq:dimensionless-equivalence}
\end{equation}
Writing the common dimensionless CDF and density as
$\widetilde F_\eta$ and $\widetilde f_\eta$, respectively, gives
\begin{align}
  F_\eta^{(j)}(t_{a,j}\tau)&=\widetilde F_\eta(\tau),
  \label{eq:dimensionless-cdf-equivalence}\\
  t_{a,j}f_\eta^{(j)}(t_{a,j}\tau)&=\widetilde f_\eta(\tau).
  \label{eq:dimensionless-density-equivalence}
\end{align}
These relations allow the numerical results to be expressed in a common
parameter space while preserving their physical time scales. For
$\alpha\leq1$, the instantaneous drift is nonnegative throughout the cycle;
for $\alpha>1$, it reverses direction during part of the cycle, producing
transient backflow despite the positive cycle mean.

%% file: tables/01_notation.tex
\begin{table*}[!tp]
  \caption{Principal Notation}
  \label{tab:notation}
  \centering
  \footnotesize
  \renewcommand{\arraystretch}{1.0}
  \setlength{\tabcolsep}{5pt}
  \begin{tabularx}{\textwidth}{@{}
    >{\raggedright\arraybackslash}p{0.28\textwidth}
    >{\raggedright\arraybackslash}X@{}}
    \toprule
    Symbol & Meaning \\
    \midrule
    \multicolumn{2}{@{}l}{\textit{Physical model and release timing}}\\
    $d,\ D$
      & Link distance; constant diffusivity.\\
    $v_0,\ a,\ \omega,\ \phi$
      & Mean drift; amplitude; angular frequency; global phase offset.\\
    $u,\ s,\ t$
      & Absolute time; release epoch; elapsed time ($u=s+t$).\\
    $S,\ Y$
      & Random release and absolute arrival times.\\
    $P,\ \eta(s),\ \mathbb T$
      & Period $2\pi/\omega$; release phase
        $(\omega s+\phi)\bmod 2\pi$; phase circle.\\
    $v_\eta(t),\ M_\eta(t)$
      & Post-release drift and cumulative displacement.\\
    $W(t),\ X_s(t),\ X_\eta(t)$      & Standard Brownian motion; position indexed by release epoch or phase.\\
    \midrule
    \multicolumn{2}{@{}l}{\textit{Conditional channel and scaling}}\\
    $T_s,\ T_\eta$
      & FHT indexed by release epoch or phase.\\
    $\mu_\eta,\ f_\eta,\ F_\eta$
      & Phase-conditioned FHT law, density, and CDF.\\
    $Q(s,A),\ p_{Y\mid S}(y\mid s)$
      & Release-to-arrival transition kernel and density.\\
    $q,\ \mu_q,\ f_q,\ F_q$
      & Protocol-induced phase law; mixture delay law, density, and CDF.\\
    $\overline f,\ \overline F$
      & Density and CDF under uniform phase mixing.\\
    $t_a,\ \tau,\ \mathcal T_\eta$
      & Advection time $d/v_0$; scaled time $t/t_a$ and FHT $T_\eta/t_a$.\\
    $\alpha,\ \Omega,\ \mathrm{Pe}$
      & Dimensionless amplitude $a/v_0$, frequency $\omega t_a$,
        and P\'eclet number $v_0d/D$.\\
    $\widetilde X_\eta,\ B,\ m_\eta$
      & Scaled position, standard Brownian motion, and cumulative displacement.\\
    $\widetilde f_\eta,\ \widetilde F_\eta$
      & Dimensionless FHT density and CDF.\\
    \midrule
    \multicolumn{2}{@{}l}{\textit{Exact characterization and C-IG approximation}}\\
    $b_\eta(t),\ \beta_\eta(\tau)$
      & Moving boundary $[d-M_\eta(t)]/\sqrt{2D}$ and its scaled form.\\
    $p_t(x),\ h_\eta(t),\ K_\eta(t,u)$
      & Brownian transition density; Volterra forcing and kernel.\\
    $\widetilde h_\eta,\ \widetilde K_\eta,\
     \widetilde{\mathcal K}_\eta$
      & Dimensionless Volterra forcing, kernel, and integral operator.\\
    $f_{\mathrm{IG}},\ F_{\mathrm{IG}}$
      & Constant-drift IG density and CDF.\\
    $L_\eta(t),\ f_{\mathrm{BM}}(t)$
      & Girsanov likelihood ratio; drift-free FHT density.\\
    $E_\eta(t),\ \Gamma_\eta(t)$
      & Cumulative-drift exponential core; omitted drift-variation term.\\
    $\ell_D(t),\ \delta_\eta(t)$
      & Diffusion scale $\sqrt{2Dt}$; relative drift deviation
        $[v_\eta(t)-v_0]/v_0$.\\
    $F_{\mathrm{mean},\eta},\
     \mathcal F_{\eta,t},\
     F_{\mathrm{smooth},\eta}$
      & Effective coordinate; auxiliary Gaussian coordinate;
        positive baseline-preserving prefactor (not CDFs).\\
    $\mathcal P_\rho(z)$
      & Gaussian positive-part mean
        $\mathbb E[(z+\rho\xi)_+]$, $\xi\sim\mathcal N(0,1)$.\\
    $\widetilde g_\eta,\ Z_\eta,\ g_\eta,\ G_\eta$
      & Unnormalized C-IG kernel, total mass, normalized density, and CDF.\\
    \midrule
    \multicolumn{2}{@{}l}{\textit{Numerical evaluation and phase diagnostics}}\\
    $\Delta t,\ N_t$
      & Time step and number of steps.\\
    $T_{\mathrm h},\ \tau_{\mathrm h},\
     t_{\max},\ t_{\mathrm{cmp}},\ t_Z$
      & Analysis, scaled analysis, evaluation, comparison,
        and normalization horizons.\\
    $\widehat f_\eta,\ \widehat F_\eta,\
     \widehat Z_\eta,\ \widehat g_\eta$
      & Numerical reference density/CDF; C-IG mass and normalized density.\\
    $\widehat v_\eta^{\mathrm d},\ J_\eta(v)$
      & Drift-only IG fit and grid-based CDF objective
        (physical $d,D$ fixed).\\
    $\mathcal D_{t_{\mathrm{cmp}}},\
     \mathcal L^1_{t_{\mathrm{cmp}}}$
      & Supremum CDF error and $L^1$ density error on $[0,t_{\mathrm{cmp}}]$.\\
    $\mathcal D_{\max},\ \mathcal L^1_{\max},\ P_{\mathrm{match}}$
      & Maximum sampled-phase errors; peak-count agreement fraction.\\
    $\epsilon_{\mathrm{mass},\eta},\
     \epsilon_{-,\eta}$
      & Raw Volterra terminal-mass discrepancy and integrated negative part.\\
    $R_{\mathrm{between}},\
     \widehat\Delta_{\mathrm{phase},48}$
      & Between-phase variance fraction; maximum pairwise CDF gap
        on the 48-phase/time grid.\\
    \bottomrule
  \end{tabularx}
\end{table*}

%% file: sections/03_exact_characterization.tex
\section{Exact First-hitting Characterization}
\label{sec:volterra-characterization}

The phase-conditioned density $f_\eta$ specifies the
release-to-arrival transition law introduced in
Section~\ref{sec:phase-resolved-channel}.
To characterize this density, we transform the absorbing-link
problem into Brownian first passage through a moving boundary
and apply classical first-passage
identities~\cite{Fortet1943,Durbin1985,Buonocore1987,Peskir2002}.
The resulting Volterra equation retains the first-contact
constraint and provides the exact characterization underlying
the numerical benchmark in our conference
manuscript~\cite{Lee2026CIG}.
Here, we develop its mathematical properties and their
implications for the phase-resolved channel, with proofs
and supporting bounds collected in
Appendices~\ref{app:volterra-proof}--%
\ref{app:stationary-ig-proof}.

A related integral-equation approach was developed by
Zoofaghari et al.~\cite{Zoofaghari2021}, who used Green's
second identity to obtain a frequency-domain boundary integral
equation for diffusion \textit{without drift} in complex geometries
with reactive boundaries.
Both approaches represent boundary effects using known
propagators, but address different difficulties:
their formulation accommodates spatial complexity, whereas
ours adopts a one-dimensional absorbing geometry to resolve
the first-hitting law under \textit{time-varying drift}.
Their published frequency-domain procedure requires
reformulation for pulsatile drift, which generally couples
frequencies; removing the drift by a coordinate transformation
instead produces a moving boundary.
The time-domain Volterra representation used here retains
the resulting dependence on both time arguments.

\subsection{Moving-boundary representation}
\label{subsec:moving-boundary}

Fix $\eta\in\Torus$ and retain $v_\eta$ and $M_\eta$ from
Section~\ref{sec:phase-resolved-channel}. Define
\begin{equation}
  b_\eta(t)=\frac{d-M_\eta(t)}{\sqrt{2D}}.
  \label{eq:physical-boundary}
\end{equation}
The trajectory in~\eqref{eq:integrated-particle-process} satisfies
\begin{equation}
  X_\eta(t)-d
  =\sqrt{2D}\,[W(t)-b_\eta(t)].
  \label{eq:pathwise-boundary-identity}
\end{equation}
Since the paths are continuous and $b_\eta(0)>0$,
\begin{equation}
  \begin{aligned}
    T_\eta
    &=\inf\{t>0:W(t)=b_\eta(t)\}\\
    &=\inf\{t>0:W(t)\geq b_\eta(t)\}.
  \end{aligned}
  \label{eq:curved-boundary-fht}
\end{equation}
Thus, fixed-receiver absorption under pulsatile drift becomes
standard Brownian first passage through a moving boundary.
Its derivatives are
\begin{equation}
  \begin{aligned}
    b_\eta'(t)&=-\frac{v_\eta(t)}{\sqrt{2D}},\\
    b_\eta''(t)&=
      -\frac{a\omega\cos(\eta+\omega t)}{\sqrt{2D}}.
  \end{aligned}
  \label{eq:boundary-derivatives}
\end{equation}
Cumulative drift determines the boundary position, while
instantaneous drift and its variation determine the slope
and curvature. The representation remains valid during
transient backflow.

\subsection{Fortet identities and the Volterra equation}
\label{subsec:fortet-volterra}

The smooth boundary and its positive initial separation imply
a continuous first-passage density on $(0,\infty)$
\cite{Peskir2002,Lehmann2002}. Let
\begin{equation}
  p_t(x)=\frac{1}{\sqrt{2\pi t}}
  \exp\!\left(-\frac{x^2}{2t}\right),
  \qquad t>0,
  \label{eq:brownian-transition-density}
\end{equation}
Hereafter, $\varphi$ and $\Phi$ denote the standard normal
density and CDF, respectively, and
$\overline\Phi=1-\Phi$ denotes the corresponding upper-tail
function. The density symbol $\varphi$ is distinct from the
global phase offset $\phi$ in the drift waveform.
Decomposition at the first boundary contact gives the Fortet identities
\begin{equation}
  \overline\Phi\!\left(\frac{b_\eta(t)}{\sqrt t}\right)
  =
  \int_0^t
    \overline\Phi\!\left(
      \frac{b_\eta(t)-b_\eta(u)}{\sqrt{t-u}}
    \right)f_\eta(u)\,du
  \label{eq:fortet-tail-identity}
\end{equation}
and
\begin{equation}
  p_t\!\left(b_\eta(t)\right)
  =
  \int_0^t
    p_{t-u}\!\left(b_\eta(t)-b_\eta(u)\right)
    f_\eta(u)\,du.
  \label{eq:fortet-density-identity}
\end{equation}
Every unrestricted Brownian path above the boundary at time
$t$ must have first crossed it at some $u\leq t$; these
identities integrate over that first-contact time.
For a compact domain $I$, $C(I)$ denotes the space of
continuous real-valued functions on $I$; on an interval,
$C^k(I)$ denotes functions with continuous derivatives
through order $k$.
Combining them yields the following second-kind formulation.

\begin{theorem}[Phase-resolved Volterra characterization]
\label{thm:phase-resolved-volterra}
For every $\eta\in\Torus$,
\begin{equation}
  f_\eta(t)
  =h_\eta(t)+\int_0^t K_\eta(t,u)f_\eta(u)\,du,
  \qquad t>0,
  \label{eq:volterra-second-kind}
\end{equation}
where
\begin{equation}
  h_\eta(t)
  =
  \left[\frac{b_\eta(t)}{t}-b_\eta'(t)\right]
  p_t\!\left(b_\eta(t)\right)
  \label{eq:volterra-forcing}
\end{equation}
and
\begin{equation}
  \begin{aligned}
    K_\eta(t,u)
    &=
    \left[
      b_\eta'(t)
      -\frac{b_\eta(t)-b_\eta(u)}{t-u}
    \right]\\
    &\quad\times
    p_{t-u}\!\left(b_\eta(t)-b_\eta(u)\right),
    \qquad u<t.
  \end{aligned}
  \label{eq:volterra-kernel}
\end{equation}
For every finite $T_{\mathrm h}>0$, the extensions
\begin{equation}
  h_\eta(0)=f_\eta(0)=0,
  \qquad K_\eta(t,t)=0
  \label{eq:volterra-continuous-extensions}
\end{equation}
make $h_\eta,f_\eta\in C([0,T_{\mathrm h}])$ and
$K_\eta\in C(\Delta_{T_{\mathrm h}})$, where
$\Delta_{T_{\mathrm h}}
=\{(t,u):0\leq u\leq t\leq T_{\mathrm h}\}$.
The density $f_\eta$ is the unique continuous solution of
\eqref{eq:volterra-second-kind} on this horizon.
\end{theorem}

\begin{IEEEproof}
See Appendix~\ref{app:volterra-proof}.
\end{IEEEproof}

For evaluation in physical variables, write
$\Delta M_\eta(t,u)=M_\eta(t)-M_\eta(u)$ and define
\begin{equation}
  \overline v_\eta(u,t)
  =\frac{\Delta M_\eta(t,u)}{t-u},
  \qquad u<t.
  \label{eq:interval-average-drift}
\end{equation}
Then
\begin{equation}
  \begin{aligned}
    h_\eta(t)
    &=
    \frac{d-M_\eta(t)+t v_\eta(t)}
         {\sqrt{4\pi Dt^3}}\\
    &\quad\times
    \exp\!\left[-\frac{(d-M_\eta(t))^2}{4Dt}\right],
  \end{aligned}
  \label{eq:physical-volterra-forcing}
\end{equation}
and
\begin{equation}
  \begin{aligned}
    K_\eta(t,u)
    &=
    \frac{\overline v_\eta(u,t)-v_\eta(t)}
         {\sqrt{4\pi D(t-u)}}\\
    &\quad\times
    \exp\!\left[
      -\frac{\Delta M_\eta(t,u)^2}{4D(t-u)}
    \right].
  \end{aligned}
  \label{eq:physical-volterra-kernel}
\end{equation}
The forcing depends on the boundary position and tangent at
$t$. The integral term accounts for earlier first contacts,
weighted by the mismatch between current and interval-average
drift. Their sum enforces the first-arrival condition.
Individually, $h_\eta$ need not be nonnegative or normalized,
and $K_\eta$ can change sign.

\subsection{Kernel regularity and dimensionless structure}
\label{subsec:kernel-regularity}

\begin{lemma}[Curvature-controlled kernel bound]
\label{lem:curvature-kernel-bound}
Let $b\in C^2([0,T_{\mathrm h}])$, let $K_b$ be the kernel
in~\eqref{eq:volterra-kernel} with $b_\eta$ replaced by $b$,
and set $L_b=\sup_{0\leq r\leq T_{\mathrm h}}|b''(r)|$.
Then
\begin{equation}
  |K_b(t,u)|
  \leq\frac{L_b}{2\sqrt{2\pi}}\sqrt{t-u},
  \qquad 0\leq u<t\leq T_{\mathrm h}.
  \label{eq:general-kernel-curvature-bound}
\end{equation}
For the sinusoidal channel, this gives
\begin{equation}
  |K_\eta(t,u)|
  \leq\frac{a\omega}{4\sqrt{\pi D}}\sqrt{t-u}.
  \label{eq:physical-kernel-curvature-bound}
\end{equation}
\end{lemma}

\begin{IEEEproof}
See Appendix~\ref{app:kernel-bound-proof}.
\end{IEEEproof}

The secant--tangent cancellation therefore removes the
diagonal singularity of the Gaussian transition density,
giving $K_\eta(t,u)\to0$ as $u\uparrow t$.

Under the scaling in Section~\ref{sec:phase-resolved-channel}, the corresponding boundary is
\begin{equation}
  \beta_\eta(\tau)
  =\frac{b_\eta(t_a\tau)}{\sqrt{t_a}}
  =\sqrt{\frac{\mathrm{Pe}}{2}}\,
    [1-m_\eta(\tau)].
  \label{eq:dimensionless-boundary}
\end{equation}
The density of $\mathcal T_\eta$ satisfies
\begin{equation}
  \widetilde f_\eta(\tau)
  =
  \widetilde h_\eta(\tau)
  +\int_0^\tau
    \widetilde K_\eta(\tau,\xi)
    \widetilde f_\eta(\xi)\,d\xi,
  \label{eq:dimensionless-volterra-equation}
\end{equation}
with
\begin{equation}
  \begin{aligned}
    \widetilde h_\eta(\tau)&=t_a h_\eta(t_a\tau),\\
    \widetilde K_\eta(\tau,\xi)
      &=t_a K_\eta(t_a\tau,t_a\xi).
  \end{aligned}
  \label{eq:dimensionless-volterra-coefficients}
\end{equation}
Since
\begin{equation}
  \beta_\eta''(\tau)
  =-\sqrt{\frac{\mathrm{Pe}}{2}}\,
    \alpha\Omega\cos(\eta+\Omega\tau),
  \label{eq:dimensionless-boundary-curvature}
\end{equation}
the kernel bound becomes
\begin{equation}
  |\widetilde K_\eta(\tau,\xi)|
  \leq
  \frac{\alpha\Omega\sqrt{\mathrm{Pe}}}{4\sqrt{\pi}}
  \sqrt{\tau-\xi}.
  \label{eq:dimensionless-kernel-curvature-bound}
\end{equation}
On $C([0,\tau_{\mathrm h}])$, use the supremum norm
$\|g\|_\infty:=\sup_{0\leq\tau\leq\tau_{\mathrm h}}|g(\tau)|$.
The dimensionless Volterra operator is
$(\widetilde{\mathcal K}_\eta g)(\tau)
:=\int_0^\tau\widetilde K_\eta(\tau,\xi)g(\xi)\,d\xi$.
Its induced operator norm is
\begin{equation}
  \|\widetilde{\mathcal K}_\eta\|_{\infty\to\infty}
  :=\sup_{\substack{g\in C([0,\tau_{\mathrm h}])\\
                   \|g\|_\infty\leq1}}
       \|\widetilde{\mathcal K}_\eta g\|_\infty.
\end{equation}
The subscript $\infty\to\infty$ indicates that both the
input and output are measured in the supremum norm.
Integration of the kernel bound gives
\begin{equation}
  \|\widetilde{\mathcal K}_\eta\|_{\infty\to\infty}
  \leq
  \frac{\alpha\Omega\sqrt{\mathrm{Pe}}}{6\sqrt{\pi}}
  \tau_{\mathrm h}^{3/2}.
  \label{eq:dimensionless-operator-bound}
\end{equation}
This bound is uniform over release phase and controls the
finite-horizon Volterra integral operator; its derivation is in
Appendix~\ref{app:kernel-bound-proof}.

\subsection{Properness and moment identities}
\label{subsec:properness-moments}

The oscillatory displacement remains bounded:
\begin{equation}
  \begin{aligned}
    R_\eta(t)&:=M_\eta(t)-v_0t\geq-C_\eta,\\
    C_\eta&=\frac{a}{\omega}(1-\cos\eta)
             \leq\frac{2a}{\omega}.
  \end{aligned}
  \label{eq:pulsatile-displacement-bound}
\end{equation}
On the same Brownian path, define
\begin{equation}
  T_\eta^+
  =
  \inf\{t>0:v_0t+\sqrt{2D}\,W(t)=d+C_\eta\}.
  \label{eq:comparison-hitting-time}
\end{equation}
At $T_\eta^+$, the pulsatile trajectory has reached at least
$d$, so continuity gives $T_\eta\leq T_\eta^+$.
This comparison with a positive-drift IG hitting time yields
the following result.

\begin{proposition}[Properness and exponential moments]
\label{prop:properness}
Under the assumptions of Section~\ref{sec:phase-resolved-channel},
$\Prob(T_\eta<\infty)=1$ for every $\eta\in\Torus$.
For $0\leq\theta\leq v_0^2/(4D)$,
\begin{equation}
  \mathbb E[e^{\theta T_\eta}]
  \leq
  \exp\!\left\{
    \frac{d+C_\eta}{2D}
    \left[v_0-\sqrt{v_0^2-4D\theta}\right]
  \right\}.
  \label{eq:exponential-moment-bound}
\end{equation}
All polynomial moments are finite, with bounds uniform
over release phase.
\end{proposition}

\begin{IEEEproof}
See Appendix~\ref{app:properness-proof}.
\end{IEEEproof}

Consequently,
\begin{equation}
  \int_0^\infty f_\eta(t)\,dt=1
  \qquad\text{for every }\eta\in\Torus.
  \label{eq:fht-unit-mass}
\end{equation}
Transient backflow redistributes arrival probability across
time without creating permanent nonarrival. The comparison
also applies to continuous periodic drift with positive
cycle mean.

The hitting condition and stopped Brownian moment identities
further imply
\begin{equation}
  \int_0^\infty
    [d-M_\eta(t)]f_\eta(t)\,dt=0
  \label{eq:first-stopped-moment-identity}
\end{equation}
and
\begin{equation}
  \int_0^\infty
    \left([d-M_\eta(t)]^2-2Dt\right)
    f_\eta(t)\,dt=0.
  \label{eq:second-stopped-moment-identity}
\end{equation}
Their integrability and stopping-time justification are given
in Appendix~\ref{app:properness-proof}. Together with the
Fortet identity, they provide checks on numerical
channel-law construction.

\subsection{Constant-drift inverse-Gaussian recovery}
\label{subsec:ig-recovery}

\begin{corollary}[Vanishing kernel and IG recovery]
\label{cor:stationary-ig-recovery}
For $b\in C^1([0,T_{\mathrm h}])$ with $T_{\mathrm h}>0$,
\begin{equation}
  K_b\equiv0
  \quad\Longleftrightarrow\quad
  b\text{ is affine},
  \label{eq:zero-kernel-iff-affine}
\end{equation}
where the kernel identity holds for all
$0\leq u<t\leq T_{\mathrm h}$.
Within the sinusoidal family with $\omega>0$, this occurs
if and only if $a=0$. All phase-conditioned densities
then reduce to
\begin{equation}
  f_{\mathrm{IG}}(t)
  =
  \frac{d}{\sqrt{4\pi Dt^3}}
  \exp\!\left[-\frac{(d-v_0t)^2}{4Dt}\right],
  \qquad t>0.
  \label{eq:stationary-ig-recovery}
\end{equation}
\end{corollary}

\begin{IEEEproof}
See Appendix~\ref{app:stationary-ig-proof}.
\end{IEEEproof}

For nonconstant drift, a nonzero kernel identifies a
survival-history correction; it does not by itself exclude an
IG-shaped density. Section~\ref{sec:cig-approximation} constructs
C-IG for this phase-conditioned law; Section~\ref{sec:cig-accuracy}
examines which arrival-time features it preserves.

%% file: sections/04_cig_approximation.tex
\section{The Corrected Inverse-Gaussian Approximation}
\label{sec:cig-approximation}

The corrected inverse-Gaussian (C-IG) approximation of our
conference manuscript~\cite{Lee2026CIG} combines a cumulative-drift
exponential core, a Gaussian expected-positive-part prefactor,
and a phase-dependent normalizer. For the same physical parameters
and release phase as the exact law, it approximates $f_\eta$ by
$g_\eta$. The construction retains accumulated displacement and
instantaneous drift through explicit terms; its treatment of
first-contact history is assessed against the Volterra law in
Section~\ref{sec:cig-accuracy}.

\subsection{Change-of-measure representation}
\label{subsec:cig-girsanov}

Fix a release phase $\eta$ and a finite horizon
$T_{\mathrm h}>0$. On a common path space, let $X(0)=0$ and
consider the drift-free reference dynamics
\begin{equation}
  dX(t)=\sqrt{2D}\,dW_0(t)
  \label{eq:cig-reference-dynamics}
\end{equation}
under the drift-free reference measure $\mathbb P_0$.
We write $\mathbb P_\eta$ for the measure associated with
the prescribed drift at release phase $\eta$, and
$\mathbb E_0$ for expectation under $\mathbb P_0$.
The filtration $\{\mathcal F_t\}_{t\geq0}$ represents the
trajectory information available up to time $t$.
The deterministic function $v_\eta$ is bounded on
$[0,T_{\mathrm h}]$, so Novikov's condition holds.
Girsanov's theorem gives the likelihood ratio
\begin{equation}
  \begin{aligned}
    L_\eta(t)
    &:=
    \left.
      \frac{d\mathbb P_\eta}{d\mathbb P_0}
    \right|_{\mathcal F_t}\\
    &=
    \exp\!\left\{
      \frac{1}{2D}\int_0^t v_\eta(u)\,dX(u)
      -\frac{1}{4D}\int_0^t v_\eta(u)^2\,du
    \right\}.
  \end{aligned}
  \label{eq:cig-girsanov-likelihood}
\end{equation}
Under $\mathbb P_\eta$, the coordinate process follows the
phase-conditioned dynamics in~\eqref{eq:released-particle-sde}.

Let $T=\inf\{t>0:X(t)=d\}$ be the first-contact functional on
this path space. Under $\mathbb P_0$, its density is
\begin{equation}
  f_{\mathrm{BM}}(t)
  =
  \frac{d}{\sqrt{4\pi Dt^3}}
  \exp\!\left(-\frac{d^2}{4Dt}\right),
  \qquad t>0.
  \label{eq:cig-drift-free-density}
\end{equation}
The change of measure implies
\begin{equation}
  f_\eta(t)
  =
  f_{\mathrm{BM}}(t)\,
  \mathbb E_0[L_\eta(t)\mid T=t]
  \label{eq:cig-exact-density-representation}
\end{equation}
for almost every $t\in(0,T_{\mathrm h})$.
Conditioning on $T=t$ is understood through a regular
conditional distribution of first-hitting trajectories.
Since the horizon is arbitrary, the representation applies
at any finite propagation time.

Integration by parts, using $X(0)=0$ and $X(t)=d$ at first
contact, gives
\begin{equation}
  \begin{aligned}
    \log L_\eta(t)
    &=
    \frac{d\,v_\eta(t)}{2D}
    -\frac{1}{2D}\int_0^t v_\eta'(u)X(u)\,du\\
    &\quad
    -\frac{1}{4D}\int_0^t v_\eta(u)^2\,du,
    \qquad \text{on }\{T=t\}.
  \end{aligned}
  \label{eq:cig-likelihood-at-hitting}
\end{equation}
The three terms describe the endpoint contribution,
the coupling between drift variation and the intermediate
trajectory, and the accumulated squared drift.
Only the middle term depends on the path once the hitting
time is fixed.

\subsection{Cumulative-drift exponential core}
\label{subsec:cig-cumulative-core}

Following~\cite{Lee2026CIG}, we approximate the path-dependent
term using the straight-path ansatz
\begin{equation}
  \overline X_t(u)=\frac{d\,u}{t},
  \qquad 0\leq u\leq t.
  \label{eq:cig-straight-path}
\end{equation}
This substitutes a representative path into the likelihood
functional; it does not identify the exact conditional mean
or a proven most probable first-hitting trajectory.

Since
$\int_0^t u\,v_\eta'(u)\,du
=t v_\eta(t)-M_\eta(t)$, substitution gives
\begin{equation}
  \begin{aligned}
    \frac{d\,v_\eta(t)}{2D}
    &-\frac{1}{2D}
      \int_0^t v_\eta'(u)\overline X_t(u)\,du\\
    &=\frac{d\,M_\eta(t)}{2Dt}.
  \end{aligned}
  \label{eq:cig-endpoint-cancellation}
\end{equation}
Combining this expression with the exponent of
$f_{\mathrm{BM}}$ yields
\begin{equation}
  \begin{aligned}
    -\frac{d^2}{4Dt}
    &+\frac{d\,M_\eta(t)}{2Dt}
    -\frac{1}{4D}\int_0^t v_\eta(u)^2\,du\\
    &=
    -\frac{(d-M_\eta(t))^2}{4Dt}
    -\Gamma_\eta(t),
  \end{aligned}
  \label{eq:cig-exponent-decomposition}
\end{equation}
where
\begin{equation}
  \begin{aligned}
    \Gamma_\eta(t)
    &:=
    \frac{1}{4D}
    \left[
      \int_0^t v_\eta(u)^2\,du
      -\frac{M_\eta(t)^2}{t}
    \right]\\
    &=
    \frac{1}{4D}
    \int_0^t
    \left[
      v_\eta(u)-\frac{M_\eta(t)}{t}
    \right]^2du
    \geq0.
  \end{aligned}
  \label{eq:cig-drift-variation-term}
\end{equation}
C-IG makes the additional approximation of omitting
$\Gamma_\eta(t)$ and retaining the exponential core
\begin{equation}
  E_\eta(t)
  =
  \exp\!\left[
    -\frac{(d-M_\eta(t))^2}{4Dt}
  \right].
  \label{eq:cig-exponential-core}
\end{equation}
No uniform smallness assumption is imposed on $\Gamma_\eta(t)$.

The retained exponent has the constant-drift IG form with
$v_0t$ replaced by $M_\eta(t)$. It therefore preserves the
release-phase dependence of accumulated transport.
The remaining construction supplies a positive prefactor
that responds to instantaneous drift.

\subsection{Gaussian positive-part prefactor}
\label{subsec:cig-epf}

Let
\begin{equation}
  \ell_D(t)=\sqrt{2Dt},
  \qquad
  \delta_\eta(t)
  =\frac{v_\eta(t)-v_0}{v_0}
  \label{eq:cig-diffusion-scale}
\end{equation}
denote the free Brownian displacement scale and the
dimensionless drift deviation. Define the effective
prefactor coordinate
\begin{equation}
  F_{\mathrm{mean},\eta}(t)
  =
  d+\delta_\eta(t)\ell_D(t).
  \label{eq:cig-effective-coordinate}
\end{equation}
This length-valued coordinate perturbs the IG numerator $d$
on the diffusion scale.

To retain a positive contribution when the effective
coordinate becomes negative, introduce the scalar Gaussian
model
\begin{equation}
  \mathcal F_{\eta,t}
  =
  F_{\mathrm{mean},\eta}(t)+\ell_D(t)\xi,
  \qquad \xi\sim\mathcal N(0,1).
  \label{eq:cig-gaussian-coordinate}
\end{equation}
This construction is termed the expected-positive-flux (EPF)
closure in our conference work~\cite{Lee2026CIG}. The auxiliary coordinate is not the
physical probability current or the position distribution
of surviving molecules. Its scale is chosen to match free
Brownian fluctuations without introducing an additional
fitted parameter.

For $\rho>0$ and $z\in\R$, define
\begin{equation}
  \begin{aligned}
    \mathcal P_\rho(z)
    &:=
    \mathbb E[(z+\rho\xi)_+]\\
    &=
    \int_0^\infty
      \frac{y}{\rho}
      \varphi\!\left(\frac{y-z}{\rho}\right)dy\\
    &=
    z\Phi\!\left(\frac{z}{\rho}\right)
    +\rho\varphi\!\left(\frac{z}{\rho}\right),
  \end{aligned}
  \label{eq:cig-gaussian-positive-part}
\end{equation}
where $(x)_+=\max\{x,0\}$ and $\varphi$ is the standard
normal density. Equation~\eqref{eq:cig-gaussian-positive-part}
is exact for the Gaussian model in
\eqref{eq:cig-gaussian-coordinate}.

A direct positive-part expectation would also modify the
constant-drift prefactor. To preserve that baseline, set
\begin{equation}
  F_{\mathrm{smooth},\eta}(t)
  =
  d\,
  \frac{
    \mathcal P_{\ell_D(t)}
    \bigl(F_{\mathrm{mean},\eta}(t)\bigr)
  }{
    \mathcal P_{\ell_D(t)}(d)
  }.
  \label{eq:cig-baseline-preserving-prefactor}
\end{equation}
The strictly positive denominator preserves
$F_{\mathrm{smooth},\eta}(t)=d$ when $a=0$.

For finite $t>0$, the Gaussian positive-part expectation
remains positive even when
$F_{\mathrm{mean},\eta}(t)<0$.
When both effective and baseline coordinates are large
relative to $\ell_D(t)$, the Gaussian smoothing is small
and the prefactor approaches
$F_{\mathrm{mean},\eta}(t)$.

\subsection{Phase normalization and constant-drift recovery}
\label{subsec:cig-normalization}

Combining the exponential core and the EPF prefactor gives
the unnormalized (raw) C-IG kernel
\begin{equation}
  \widetilde g_\eta(t)
  =
  \frac{F_{\mathrm{smooth},\eta}(t)}
       {\sqrt{4\pi Dt^3}}
  \exp\!\left[
    -\frac{(d-M_\eta(t))^2}{4Dt}
  \right],
  \qquad t>0.
  \label{eq:cig-raw-kernel}
\end{equation}
Define its phase-dependent mass and normalized density by
\begin{equation}
  Z_\eta=\int_0^\infty\widetilde g_\eta(u)\,du,
  \qquad
  g_\eta(t)=\frac{\widetilde g_\eta(t)}{Z_\eta}.
  \label{eq:cig-normalized-density}
\end{equation}

For the positive-mean sinusoidal model,
$0<Z_\eta<\infty$. To see finiteness, first note that
$M_\eta(t)=O(t)$ and
$F_{\mathrm{smooth},\eta}(t)\to d$ as $t\downarrow0$.
Thus the raw kernel is bounded near zero by a constant
multiple of $t^{-3/2}e^{-c/t}$ for some $c>0$.
At large times, $M_\eta(t)=v_0t+O(1)$ and
$\delta_\eta(t)$ is bounded.
The scaling relation
$\mathcal P_\rho(z)=\rho\mathcal P_1(z/\rho)$
then gives $F_{\mathrm{smooth},\eta}(t)=O(1)$, and hence
\begin{equation}
  \widetilde g_\eta(t)
  =
  O\!\left(
    t^{-3/2}e^{-v_0^2t/(4D)}
  \right),
  \qquad t\to\infty.
  \label{eq:cig-tail-integrability}
\end{equation}
Strict positivity of the raw kernel gives $Z_\eta>0$.
Thus $g_\eta$ is a positive probability density.

Normalization preserves the raw kernel's relative time dependence
and peak locations. Its mass discrepancy $|Z_\eta-1|$ and the
normalized density/CDF errors measure distinct properties;
the numerical conventions are given in
\secref{sec:numerical-methods}{subsec:numerical-settings}.

When $a=0$, $M_\eta(t)=v_0t$ and
$F_{\mathrm{smooth},\eta}(t)=d$, so $Z_\eta=1$ and C-IG reduces
to the IG density in~\eqref{eq:stationary-ig-recovery}.

\subsection{Channel evaluation and relation to the exact law}
\label{subsec:cig-channel-evaluation}

Let
\begin{equation}
  G_\eta(t)=\int_0^t g_\eta(u)\,du
  \label{eq:cig-cdf}
\end{equation}
denote the C-IG CDF, and extend $g_\eta$ by zero for
$t\leq0$. Substitution into the transition law of
Section~\ref{sec:phase-resolved-channel} gives
\begin{equation}
  p_{Y\mid S}^{\mathrm{C\text{-}IG}}(y\mid s)
  =
  g_{\eta(s)}(y-s)\,
  \mathbf 1_{\{y>s\}}.
  \label{eq:cig-timing-transition}
\end{equation}
The functions $Z_\eta$, $g_\eta$, and $G_\eta$ inherit
the release-phase periodicity of $v_\eta$ and $M_\eta$.

For known physical parameters and a known periodic waveform,
the normalizers can be tabulated offline on a phase grid.
With analytical $M_\eta(t)$ and fixed-cost lookup or interpolation
of $Z_\eta$, density evaluation costs $O(1)$ at fixed numerical
precision, excluding offline normalization. CDF evaluation
requires integration or a precomputed table. The exact Volterra
density can also be tabulated, following its recursive construction
over a time grid.

The exact law enforces first contact through the Volterra
equation. C-IG replaces the conditional path calculation by
the cumulative-drift and EPF approximations, followed by
normalization; it is not obtained by discarding the Volterra
integral term. Neither the kernel bound
in~\eqref{eq:dimensionless-operator-bound} nor normalization
provides a C-IG error bound. Section~\ref{sec:cig-accuracy}
assesses its density and CDF errors against the numerical
Volterra reference.

%% file: sections/05_numerical_methods.tex
\section{Numerical Methods and Reference Validation}
\label{sec:numerical-methods}

The Volterra evaluator supplies the reference for the C-IG
comparisons inherited from our conference manuscript~\cite{Lee2026CIG}
and for the phase-family analysis developed here. We collect the
comparison conventions and reliability checks below, with implementation
settings in Appendix~\ref{app:numerical-implementation}.

\subsection{Finite-grid Volterra evaluation}
\label{subsec:volterra-numerics}

On $t_n=n\Delta t$, $n=0,\ldots,N_t$, with
$t_{N_t}=t_{\max}$, composite trapezoidal quadrature of
\eqref{eq:volterra-second-kind} gives
\begin{equation}
  \widehat f_{\eta,n}
  =h_\eta(t_n)+\Delta t\sum_{j=1}^{n-1}
    K_\eta(t_n,t_j)\widehat f_{\eta,j},\qquad n\geq1,
  \label{eq:volterra-product-integration}
\end{equation}
with $\widehat f_{\eta,0}=0$. The zero initial density and
diagonal kernel eliminate the endpoint terms, yielding an explicit
scheme~\cite{Linz1985Volterra}. The coefficients use
\eqref{eq:physical-volterra-forcing} and
\eqref{eq:physical-volterra-kernel}. Cumulative quadrature gives
\begin{equation}
  \widehat F_{\eta,n}=\widehat F_{\eta,n-1}
  +\frac{\Delta t}{2}
    (\widehat f_{\eta,n-1}+\widehat f_{\eta,n}),
  \qquad\widehat F_{\eta,0}=0.
  \label{eq:volterra-cdf-quadrature}
\end{equation}
Computing kernel rows as needed requires $O(N_t^2)$ operations
and $O(N_t)$ memory per phase. Raw densities are neither clipped nor
renormalized. The diagnostics
\begin{align}
  \epsilon_{\mathrm{mass},\eta}
    &=|1-\widehat F_\eta(t_{\max})|,
    \label{eq:terminal-mass-error}\\
  \epsilon_{-,\eta}
    &=\operatorname{Trap}_{[0,t_{\max}]}
      [\max\{-\widehat f_\eta,0\}]
    \label{eq:negative-mass-error}
\end{align}
measure terminal-mass discrepancy and the integrated negative part. The former
combines time discretization and horizon truncation.

\subsection{Comparison metrics and baselines}
\label{subsec:common-baselines-metrics}

We compare normalized C-IG with the Volterra reference using
CDF and density discrepancies on a common horizon:
\begin{align}
  \mathcal D_{t_{\mathrm{cmp}}}(F,G)
    &=\sup_{0\leq t\leq t_{\mathrm{cmp}}}|F(t)-G(t)|,
    \label{eq:finite-horizon-cdf-distance}\\
  \mathcal L^1_{t_{\mathrm{cmp}}}(f,g)
    &=\int_0^{t_{\mathrm{cmp}}}|f(t)-g(t)|\,dt.
    \label{eq:finite-horizon-density-distance}
\end{align}
The supremum is evaluated as a maximum on the common time grid
and the integral by trapezoidal quadrature. We omit the horizon
subscript when clear. CDF errors measure cumulative arrival differences;
the integrated density error measures discrepancies in the arrival profile.

The conference direct comparison also includes nominal IG with
$v=v_0$ and a drift-only fitted IG. Writing $F_{\mathrm{IG}}(t;d,D,v)$
for the CDF associated with~\eqref{eq:stationary-ig-recovery}, the latter
selects $\widehat v_\eta^{\mathrm d}>0$ by numerical minimization of
\begin{equation}
  J_\eta(v)=\max_{t_n\leq t_{\mathrm{cmp}}}
    |\widehat F_\eta(t_n)-F_{\mathrm{IG}}(t_n;d,D,v)|.
  \label{eq:ig-grid-objective}
\end{equation}
Physical $d$ and $D$ remain fixed. This retrospective benchmark
measures the effect of replacing pulsatile transport by an effective
constant drift; its reported residual concerns this restricted fit.
Density errors use the same CDF-fitted drift.

The direct C-IG comparison and regime sweep use raw cumulative
Volterra quadrature. Phase-family CDFs use cumulative maximization and
clipping to $[0,1]$, without terminal normalization; their solver
diagnostics use raw output. Only the ECDF comparison in
Fig.~\ref{fig:volterra-oracle-validation}(c) uses the terminal-normalized
reference specified in
\secref{sec:numerical-methods}{subsec:trajectory-validation}.

\subsection{Evaluation settings and C-IG normalization}
\label{subsec:numerical-settings}

Table~\ref{tab:consolidated-numerical-settings} distinguishes the
conference C-IG experiments from the two 48-phase channel families.
The latter support the phase and protocol analysis in
Section~\ref{sec:phase-structure}; they are separate parameter settings.
The non-reversing phase family is referred to as the regular family.

\begin{table*}[!tp]
\caption{Main Numerical Settings for Approximation and
Phase-Structure Comparisons}
\label{tab:consolidated-numerical-settings}
\centering
\scriptsize
\renewcommand{\arraystretch}{1.18}
\setlength{\tabcolsep}{4pt}
\begin{tabularx}{\textwidth}{@{}
  >{\RaggedRight\arraybackslash}p{0.18\textwidth}
  >{\RaggedRight\arraybackslash}p{0.28\textwidth}
  >{\RaggedRight\arraybackslash}p{0.22\textwidth}
  >{\RaggedRight\arraybackslash}X@{}}
\toprule
Numerical block
& Physical parameters $(d,D,v_0,a,\omega)$
& Release phases
& Time grid and horizons\\
\midrule
Direct C-IG comparison
& $(5,1,1,2,2\pi)$
& $\eta=0$
& $\Delta t=0.002$;
  $t_{\mathrm{cmp}}=20$;
  C-IG normalization horizon $t_Z=80$.\\
\addlinespace
C-IG operating-regime sweep
& $d=v_0=1$, $D=0.25$;
  13 equally spaced $\alpha=a/v_0\in[0,2]$;
  13 logarithmically spaced $\Omega\in[0.5,16]$
& $\eta_k=2\pi k/32$,
  $k=0,\ldots,31$
& $\Delta t=0.004$;
  $t_{\mathrm{cmp}}=12$;
  $t_Z=80$.
  Selected refinements are described in
  Section~\ref{sec:cig-accuracy}.\\
\addlinespace
Regular phase family
& $(1,0.2,0.8,0.3,2)$;
  $(\alpha,\Omega,\mathrm{Pe})=(0.375,2.5,4)$
& $\eta_k=2\pi k/48$,
  $k=0,\ldots,47$
& $\Delta t=0.002$;
  $t_{\max}=t_{\mathrm{cmp}}=12$.\\
\addlinespace
Transient-backflow phase family
& $(1,0.2,0.6,0.9,2)$;
  $(\alpha,\Omega,\mathrm{Pe})=(1.5,10/3,3)$
& $\eta_k=2\pi k/48$,
  $k=0,\ldots,47$
& $\Delta t=0.002$;
  $t_{\max}=t_{\mathrm{cmp}}=20$.\\
\bottomrule
\end{tabularx}
\end{table*}

The C-IG normalizer is computed on a separate horizon $t_Z$:
\begin{equation}
  \widehat Z_\eta(t_Z)
    =\operatorname{Trap}_{[0,t_Z]}[\widetilde g_\eta],
  \qquad
  \widehat g_\eta(t)=\frac{\widetilde g_\eta(t)}{\widehat Z_\eta(t_Z)}.
  \label{eq:numerical-cig-normalization}
\end{equation}
Quadrature uses the Volterra time spacing; cumulative integration
of $\widehat g_\eta$ gives the C-IG CDF. The raw mass
$\widehat Z_\eta$ is reported separately from normalized shape errors,
without renormalizing either law on $[0,t_{\mathrm{cmp}}]$.
For the direct comparison, extending $t_Z$ from $40$ to $80$ changes
$\widehat Z_0$ by $1.2\times10^{-5}$; the raw Volterra terminal deficit
at $t_{\mathrm{cmp}}=20$ is approximately $2.5\times10^{-3}$.

\subsection{Deterministic reference checks}
\label{subsec:deterministic-validation}

Figure~\ref{fig:volterra-oracle-validation}(a) checks the
constant-drift case $(d,D,v_0,a,\omega)=(1,0.2,0.8,0,2)$ on $[0,10]$.
At $\Delta t=0.004$, the density agrees with the analytical IG law
to machine precision, the CDF error is $3.78\times10^{-6}$, and the
terminal deficit is $4.71\times10^{-5}$. For the regular-pulsation case at
$\eta=0.4$, density differences from the $\Delta t=0.004$ solution
on $[0,10]$ decrease from $3.62\times10^{-5}$ at $\Delta t=0.016$
to $9.46\times10^{-6}$ at $\Delta t=0.008$ [panel (b)]. This factor
of $3.83$ supports refinement consistency without establishing an
asymptotic convergence order.

Table~\ref{tab:reference-diagnostics} checks the Fortet and
stopped-moment identities independently of the recursion, using the
residuals defined in
\appsubref{app:numerical-implementation}{appsub:deterministic-details}.
Their finite-grid values include diagnostic quadrature error and,
for moments, tail truncation. The selected cases have no negative
density values. All 96 phase-family laws satisfy the mass and
nonnegativity criteria in
\appsubref{app:numerical-implementation}{appsub:phase-gates}, with
maximum positive terminal deficits $1.16\times10^{-5}$ and
$1.40\times10^{-5}$. The regime sweep has the additional time and
phase refinements in
\secref{sec:cig-accuracy}{subsec:cig-regime-refinement}.

\begin{table*}[!tp]
\caption{Volterra Reference Diagnostics.
The Upper Block Uses $\Delta t=10^{-3}$;
the Lower Block Reports Componentwise Maxima over Each
48-Phase Family on the Grids in
Table~\ref{tab:consolidated-numerical-settings}.}
\label{tab:reference-diagnostics}
\centering
\scriptsize
\renewcommand{\arraystretch}{1.15}
\setlength{\tabcolsep}{5pt}
\begin{tabular}{llcccc}
\toprule
Case & Phase
& $\epsilon_{\mathrm{mass}}$
& $\max_t|r_{\mathrm F,\eta}(t)|$
& $|\widehat r_{1,\eta}|$
& $|\widehat r_{2,\eta}|$\\
\midrule
Constant-drift IG & ---
& $5.10\times10^{-7}$
& $2.56\times10^{-6}$
& $6.08\times10^{-6}$
& $6.95\times10^{-5}$\\
Regular & $0.4$
& $6.35\times10^{-7}$
& $4.69\times10^{-6}$
& $3.93\times10^{-6}$
& $4.42\times10^{-5}$\\
Regular & $0.4+\pi/2$
& $1.18\times10^{-6}$
& $3.22\times10^{-6}$
& $7.33\times10^{-6}$
& $8.68\times10^{-5}$\\
Regular & $0.4+\pi$
& $3.36\times10^{-7}$
& $1.33\times10^{-6}$
& $8.64\times10^{-6}$
& $9.62\times10^{-5}$\\
Regular & $0.4+3\pi/2$
& $7.08\times10^{-8}$
& $2.69\times10^{-6}$
& $4.74\times10^{-6}$
& $5.05\times10^{-5}$\\
Backflow & $3\pi/2$
& $5.93\times10^{-6}$
& $3.83\times10^{-6}$
& $7.86\times10^{-5}$
& $8.86\times10^{-4}$\\
\midrule
Regular family & Maximum over 48 phases
& $1.16\times10^{-5}$
& $1.36\times10^{-5}$
& $1.06\times10^{-4}$
& $9.38\times10^{-4}$\\
Backflow family & Maximum over 48 phases
& $1.40\times10^{-5}$
& $2.59\times10^{-5}$
& $1.77\times10^{-4}$
& $1.94\times10^{-3}$\\
\bottomrule
\end{tabular}
\end{table*}

\begin{figure}[!tp]
  \centering
  \includegraphics[width=0.98\columnwidth]
    {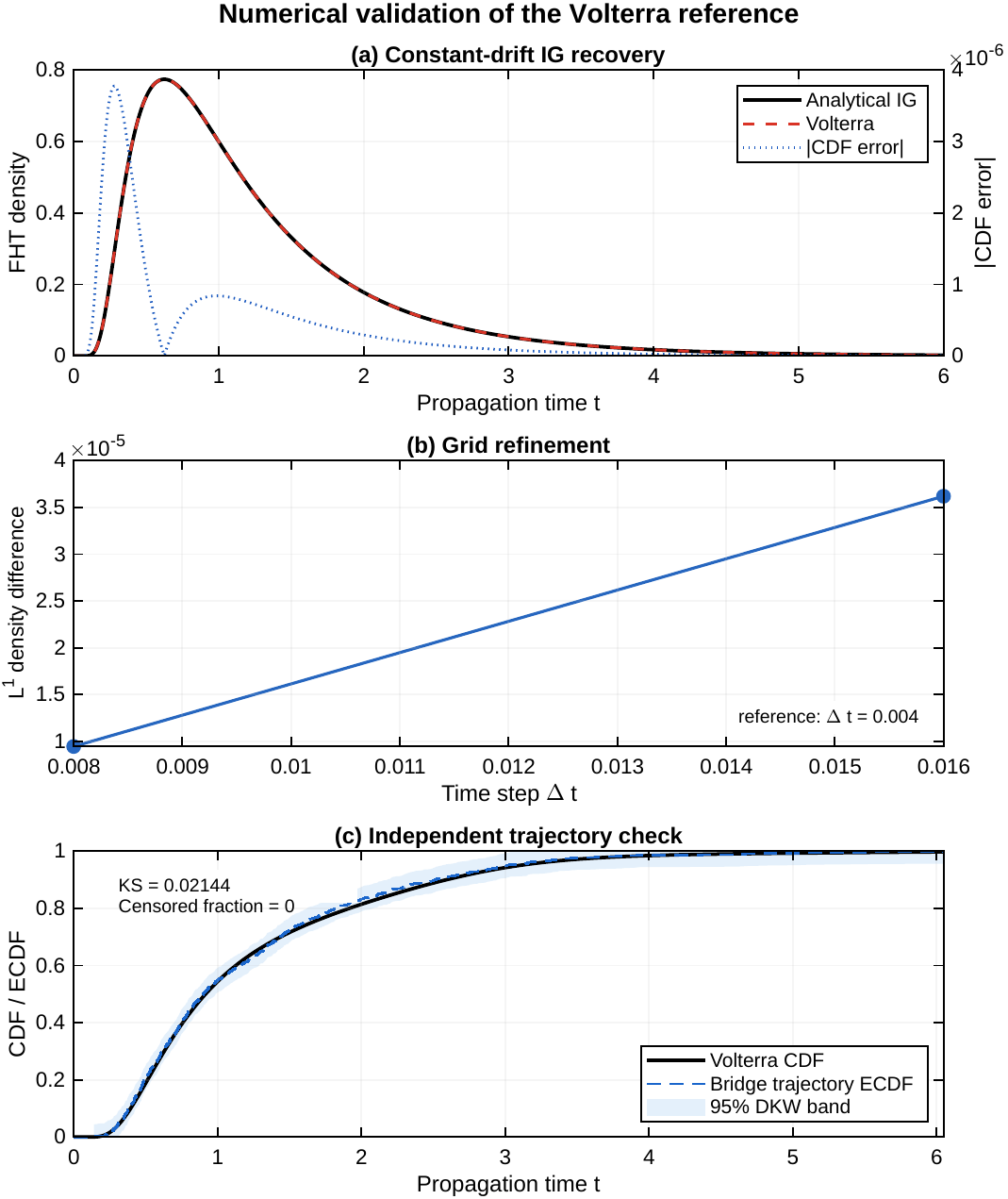}
  \caption{Validation of the finite-grid Volterra evaluator.
  (a) Analytical constant-drift IG recovery and absolute CDF
  error. (b) Pulsatile-drift density refinement relative to
  the $\Delta t=0.004$ solution.
  (c) Volterra CDF and independently simulated,
  bridge-corrected trajectory ECDF; the shaded region
  shows the nominal 95\% DKW band around the reference CDF.
  The terminal-mass normalization used only in panel (c)
  is specified in \secref{sec:numerical-methods}{subsec:trajectory-validation}.}
  \label{fig:volterra-oracle-validation}
\end{figure}

\subsection{Independent trajectory comparisons}
\label{subsec:trajectory-validation}

For the regular-pulsation case at $\eta=0.4+\pi/2$, all 1000
bridge-corrected trajectories arrive by time $15$. In
Fig.~\ref{fig:volterra-oracle-validation}(c), the Volterra CDF at
$\Delta t=0.004$ is divided by its terminal mass for comparison
with their ECDF. The KS statistic relative to the interpolated reference is $0.02144$,
within the nominal 95\% Dvoretzky--Kiefer--Wolfowitz (DKW)
confidence-band half-width $0.04295$~\cite{Massart1990}.

The conference direct comparison~\cite{Lee2026CIG} uses $10^5$
Euler--Maruyama trajectories with interpolated crossings. Its maximum
ECDF discrepancy from raw Volterra quadrature at histogram bin edges
in $[0,8]$ is $0.00415$, compared with the asymptotic one-sample
KS critical value $0.00430$ at the 5\% significance level.
This value provides a sampling-error scale for the bin-edge comparison.
These sampling-based thresholds exclude time discretization,
crossing approximation, and Volterra error. Generators, seeds,
steps, and censoring conventions are specified in
\appsubref{app:numerical-implementation}{appsub:trajectory-details}.

%% file: sections/06_cig_accuracy.tex
\Needspace{8\baselineskip}
\section{C-IG Accuracy and Operating Regimes}
\label{sec:cig-accuracy}

We assess how well C-IG preserves arrival-time structure using
the direct comparison and amplitude--frequency--phase sweep of our
conference manuscript~\cite{Lee2026CIG}. The Volterra reference and
metrics in Section~\ref{sec:numerical-methods} distinguish cumulative
agreement, density accuracy, and resolved peak structure.

\subsection{Direct phase-conditioned comparison}
\label{subsec:cig-direct-comparison}

The direct-comparison setting in
Table~\ref{tab:consolidated-numerical-settings} has
$(\alpha,\Omega,\mathrm{Pe})=(2,10\pi,5)$ at $\eta=0$,
including transient flow reversal.
Figure~\ref{fig:cig-direct-comparison} compares C-IG with the
Volterra reference, nominal and drift-only fitted IG, and particle simulations.

The estimated raw C-IG mass is
\begin{equation}
  \widehat Z_0(80)=1.151097,
  \label{eq:cig-direct-raw-mass}
\end{equation}
approximately $15.1\%$ above unity.
After normalization by~\eqref{eq:numerical-cig-normalization},
its errors on $[0,20]$ are reported in
Table~\ref{tab:cig-direct-errors}.

\begin{table}[!tp]
\caption{Direct-Comparison Errors Relative to the
Volterra Reference on $[0,20]$}
\label{tab:cig-direct-errors}
\centering
\renewcommand{\arraystretch}{1.12}
\setlength{\tabcolsep}{5pt}
\begin{tabular}{lcc}
\toprule
Approximation
& $\mathcal D_{20}$
& $\mathcal L^1_{20}$\\
\midrule
Nominal IG, $v=v_0$
& $0.08047$ & $0.43072$\\
Drift-only IG, $v=1.12889$
& $0.03443$ & $0.42486$\\
Phase-normalized C-IG
& $0.04325$ & $0.33778$\\
\bottomrule
\end{tabular}
\end{table}

C-IG has smaller density and CDF errors than nominal IG and
the smallest density error of the three approximations. The
drift-only fit at $\widehat v_{\eta=0}^{\mathrm d}=1.12889$ has the
smallest CDF error. The rankings therefore differ between cumulative
arrival probabilities and the density profile.

This distinction is visible in
Fig.~\ref{fig:cig-direct-comparison}(a).
The C-IG construction responds to both accumulated
displacement and instantaneous drift, allowing multiple
density peaks. A constant-drift IG density is unimodal,
so adjustment of its drift can change the dominant
arrival scale but cannot reproduce multiple peaks.
Panel (b) shows that the resulting local differences
can coexist with relatively close cumulative curves.

\begin{figure*}[!tp]
  \centering
  \includegraphics[width=0.98\textwidth]
    {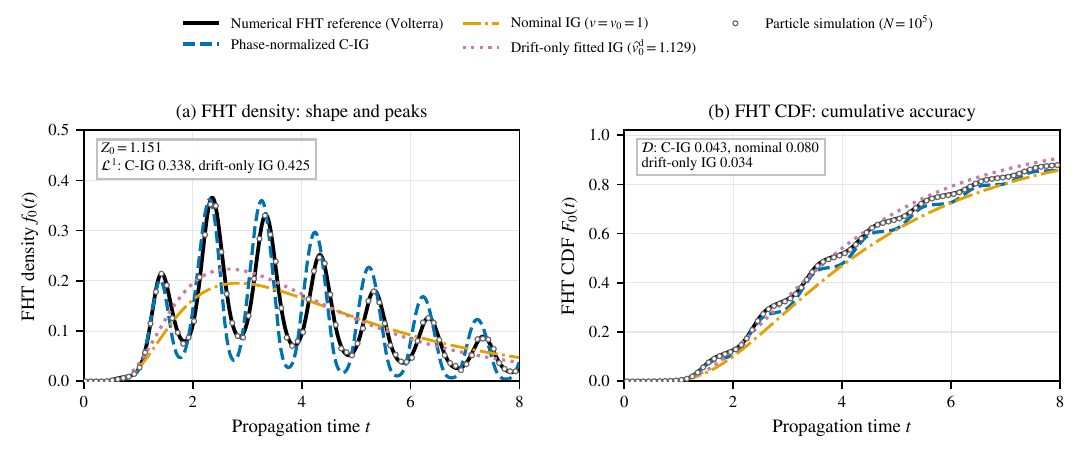}
  \caption{Phase-conditioned first-hitting distributions
for $(d,D,v_0,a,\omega)=(5,1,1,2,2\pi)$ and $\eta=0$,
adapted from~\cite{Lee2026CIG} with updated labels.
(a) Density comparison with the Volterra reference,
phase-normalized C-IG, nominal IG, drift-only fitted IG, and
$10^5$ particle trajectories.
(b) Corresponding cumulative distributions.
The panels show $[0,8]$; the reported errors use $[0,20]$,
and the C-IG normalization horizon is $80$.
The drift-only fit holds the physical $d,D$ fixed and minimizes
the grid-based CDF discrepancy defined in
\secref{sec:numerical-methods}{subsec:common-baselines-metrics}.}
  \label{fig:cig-direct-comparison}
\end{figure*}

The particle check is reported in
\secref{sec:numerical-methods}{subsec:trajectory-validation};
all approximation errors use the deterministic Volterra reference.

\subsection{Amplitude--frequency--phase assessment}
\label{subsec:cig-regime-assessment}

The sweep in Table~\ref{tab:consolidated-numerical-settings}
contains 169 amplitude--frequency cells at $\mathrm{Pe}=4$
and 5408 phase-conditioned Volterra solutions, compared on
$[0,12]$.

For each cell, define the maximum sampled-phase errors
\begin{equation}
  \begin{aligned}
    \mathcal D_{\max}(\alpha,\Omega)
    &=
    \max_{0\leq k<32}
      \mathcal D_{12}
      \bigl(F_{\eta_k},G_{\eta_k}\bigr),\\
    \mathcal L^1_{\max}(\alpha,\Omega)
    &=
    \max_{0\leq k<32}
      \mathcal L^1_{12}
      \bigl(f_{\eta_k},g_{\eta_k}\bigr).
  \end{aligned}
  \label{eq:cig-worst-sampled-phase-errors}
\end{equation}
The dependence of $F_\eta$ and $G_\eta$ on
$(\alpha,\Omega,\mathrm{Pe})$ is suppressed; the errors use
the time-grid metrics of
\secref{sec:numerical-methods}{subsec:common-baselines-metrics}.

We use three empirical accuracy classes:
\begin{equation}
  \begin{aligned}
    \text{strict:}\quad&
      \mathcal D_{\max}\leq0.025,\quad
      \mathcal L^1_{\max}\leq0.10;\\
    \text{usable only:}\quad&
      \mathcal D_{\max}\leq0.05,\quad
      \mathcal L^1_{\max}\leq0.20,\\
    &\text{but not strict};\\
    \text{caution:}\quad&
      \text{otherwise}.
  \end{aligned}
  \label{eq:cig-operating-classes}
\end{equation}
Here, ``usable'' includes both strict and usable-only
cells. These labels refer to the stated numerical
criteria, rather than to a universal communication
performance requirement.

The final map in
Fig.~\ref{fig:cig-operating-regimes}(a), including the
refinement described below, contains 43 strict,
34 usable-only, and 92 caution cells.
Table~\ref{tab:cig-regime-summary} separates the usable
counts according to the sign of the instantaneous drift.

\begin{table}[!tp]
\caption{Cells Meeting the Usable Criteria at
$\mathrm{Pe}=4$}
\label{tab:cig-regime-summary}
\centering
\renewcommand{\arraystretch}{1.12}
\setlength{\tabcolsep}{5pt}
\begin{tabular}{lcc}
\toprule
Amplitude range
& Usable / total
& Fraction\\
\midrule
$\alpha<1$
& $69/78$ & $88.5\%$\\
$\alpha=1$
& $6/13$ & $46.2\%$\\
$\alpha>1$
& $2/78$ & $2.6\%$\\
\bottomrule
\end{tabular}
\end{table}

Most sampled cells with strictly positive drift satisfy
the usable criteria.
The fraction decreases at $\alpha=1$, where the waveform
touches zero, and decreases further when $\alpha>1$
produces intervals of reverse drift.
Nevertheless, some low-frequency cells with
$\alpha<1$ also exceed the thresholds.
The no-reversal condition alone is therefore insufficient
to establish the prescribed accuracy.

\subsection{Time and phase refinement}
\label{subsec:cig-regime-refinement}

Refining 12 selected cells from the base grid in
Table~\ref{tab:consolidated-numerical-settings} uses
$\Delta t=0.002$ and 64--128 phases.

Across these checks, the maximum sampled-phase CDF
and density errors change by at most
$3.03\times10^{-4}$ and $3.16\times10^{-4}$,
respectively.
Eleven of the twelve classifications remain unchanged.

The remaining cell,
\begin{equation}
  (\alpha,\Omega)=\left(\frac23,2^{1/4}\right),
  \label{eq:cig-refined-boundary-cell}
\end{equation}
lies near the usable CDF threshold.
A local phase search with $\Delta t=10^{-3}$,
followed by confirmation at $\Delta t=5\times10^{-4}$,
gives
\begin{equation}
  \mathcal D_{12}=0.050038
  \quad\text{at}\quad
  \frac{\eta}{2\pi}\approx0.526519.
  \label{eq:cig-refined-boundary-error}
\end{equation}
Because this exceeds $0.05$, the cell is assigned to
the caution class in the final map.
Only this cell's classification is updated.

This threshold crossing illustrates phase-resolution sensitivity.
The final map summarizes sampled phases and selected local
refinements, without certifying a continuous-phase supremum.

\subsection{Density peaks and transient backflow}
\label{subsec:cig-peak-counts}

We separately examine whether C-IG reproduces the number
of resolved density peaks.
For each phase, peaks are detected on the comparison
window using a height threshold of $1\%$ and a prominence
threshold of $2\%$ of that density's phasewise maximum.
Each density uses its own maximum to set the thresholds.

Let $N_{\mathrm{pk}}(f_\eta)$ and
$N_{\mathrm{pk}}(g_\eta)$ denote the resulting counts.
For each amplitude--frequency cell, define
\begin{equation}
  P_{\mathrm{match}}(\alpha,\Omega)
  =
  \frac{1}{32}
  \sum_{k=0}^{31}
  \mathbf 1_{\{
    N_{\mathrm{pk}}(f_{\eta_k})
    =
    N_{\mathrm{pk}}(g_{\eta_k})
  \}}.
  \label{eq:cig-peak-count-agreement}
\end{equation}
Figure~\ref{fig:cig-operating-regimes}(b) shows this
fraction on the original 32-phase grid.
The local refinement of the classification map does not
replace the phase grid used for the peak-count panel.

Averaging $P_{\mathrm{match}}$ over the cells in each
amplitude range gives $78.5\%$ for $\alpha<1$,
$60.8\%$ for $\alpha=1$, and $55.4\%$ for $\alpha>1$.
Matching counts do not imply matching peak locations or
heights, but mismatched counts identify a structural
difference that a CDF metric can obscure.

\begin{figure*}[!tp]
  \centering
  \includegraphics[width=0.98\textwidth]
    {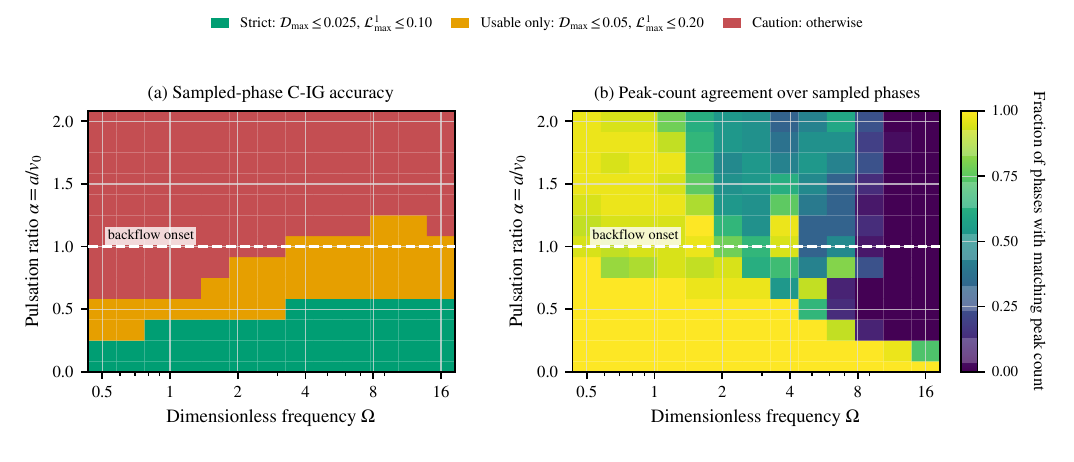}
  \caption{C-IG accuracy and density structure at
$\mathrm{Pe}=4$, adapted from~\cite{Lee2026CIG} with updated labels.
(a) Accuracy classes based on maximum sampled-phase
CDF and density errors over $[0,12]$, retaining the
refined boundary-cell classification described in
\secref{sec:cig-accuracy}{subsec:cig-regime-refinement}.
(b) Fraction of the original 32 release phases with
matching density peak counts under the $1\%$ height
and $2\%$ prominence rule.
The dashed line marks $\alpha=1$, the boundary between
nonnegative drift and transient flow reversal.}
  \label{fig:cig-operating-regimes}
\end{figure*}

Reverse-flow intervals alter which trajectories survive to later
favorable portions of the waveform. The loss of accuracy is
consistent with the EPF closure's treatment of survival history
described in
\secref{sec:cig-approximation}{subsec:cig-channel-evaluation},
although this map does not isolate the contribution of each
approximation step.

The operating classes describe the tested sinusoidal waveforms
at $\mathrm{Pe}=4$. Other P\'eclet numbers, waveforms, or
tail-sensitive windows require their own comparisons.
Section~\ref{sec:phase-structure} examines the phase family
and its protocol-induced mixtures, connecting these arrival-time
features to timing likelihoods and reception deadlines.

%% file: sections/07_phase_structure.tex
\begin{figure*}[!tp]
  \centering
  \includegraphics[width=0.485\textwidth]
    {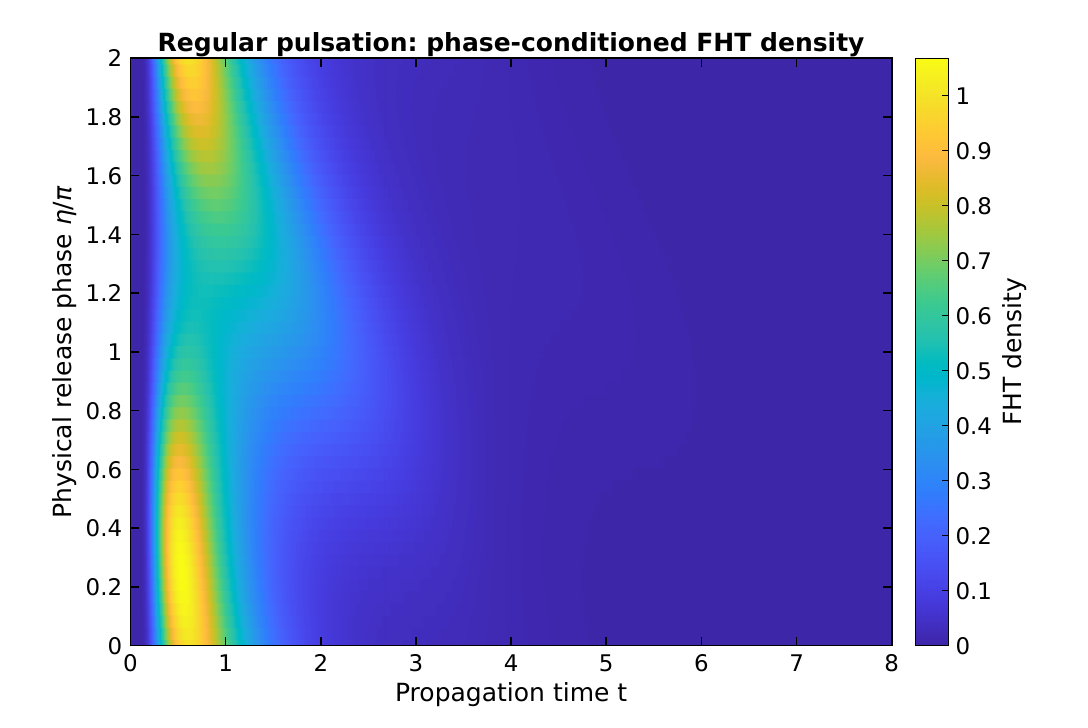}
  \hfill
  \includegraphics[width=0.485\textwidth]
    {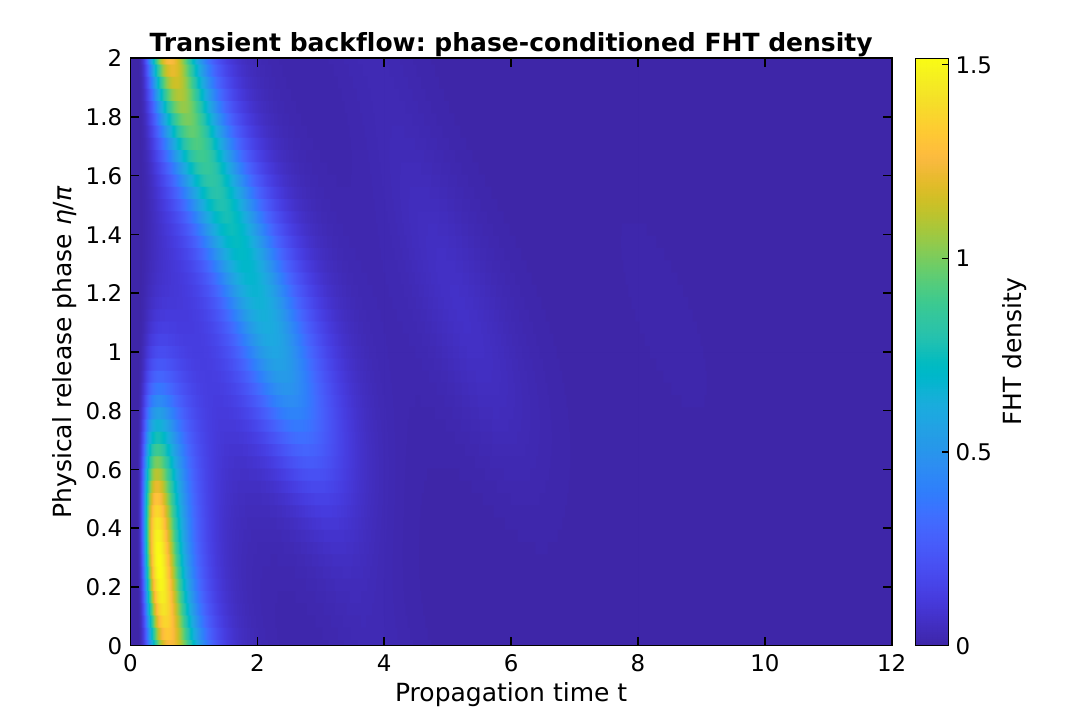}
  \caption{Phase--time density surfaces of the validated
  channel families.
  Left: regular pulsation with
  $(\alpha,\Omega,\mathrm{Pe})=(0.375,2.5,4)$.
  Right: transient backflow with
  $(\alpha,\Omega,\mathrm{Pe})=(1.5,10/3,3)$.
  Computation uses 48 endpoint-excluded release phases;
  the periodic endpoint is added only for display.
  The panels show $t\in[0,8]$ and $[0,12]$, respectively,
  using separate color scales; the computation horizons are
  $12$ and $20$ as listed in
  Table~\ref{tab:consolidated-numerical-settings}.}
  \label{fig:phase-time-channel-surfaces}
\end{figure*}

\Needspace{12\baselineskip}
\section{Phase-resolved Channel Structure and Communication Implications}
\label{sec:phase-structure}

The C-IG comparisons identify which features of a conditional
arrival law survive approximation. We now examine the phase family
itself: how the arrival profile changes with release phase, what a
protocol-induced mixture retains, and which law enters communication
calculations. The two 48-phase families in
Table~\ref{tab:consolidated-numerical-settings} provide the numerical
examples.

\subsection{Phase--time channel structure}
\label{subsec:phase-resolved-deformation}

Figure~\ref{fig:phase-time-channel-surfaces} shows the
conditional densities over release phase and propagation time.
For regular pulsation, the principal arrival ridge shifts and changes
shape with $\eta$. The backflow case exhibits a wider displacement,
including delayed arrivals at adverse release phases. The two parameter settings
differ in amplitude, dimensionless frequency, and P\'eclet number,
so the comparison does not isolate backflow as the sole cause.

Phase changes both accumulated displacement $M_\eta(t)$ and
the moving boundary $b_\eta(t)$. An early favorable interval
removes molecules before later parts of the waveform are encountered,
changing the population available for subsequent peaks. This survival
history is present in the Volterra law and helps explain why the
C-IG assessment must examine density structure as well as cumulative
arrivals.

\subsection{Protocol-induced mixtures and phase averaging}
\label{subsec:phase-averaging-concealment}

Uniform release-phase randomization gives the finite-grid
mixture
\begin{equation}
  \widehat{\overline F}_{48}(t)
  =\frac{1}{48}\sum_{k=0}^{47}\widehat F_{\eta_k}(t),
  \qquad\eta_k=\frac{2\pi k}{48}.
  \label{eq:forty-eight-phase-average-cdf}
\end{equation}
This approximates~\eqref{eq:uniform-phase-mixture} using the
CDF convention in
\secref{sec:numerical-methods}{subsec:common-baselines-metrics}.
It combines the phase-dependent arrival profiles but no longer identifies
the conditional law associated with a particular release epoch.

To quantify the variation being averaged, set $\Theta=\eta(S)$
and $T=Y-S$. The law of total variance gives
\begin{equation}
  \begin{aligned}
    \operatorname{Var}(T)
    &=\mathbb E_\Theta[\operatorname{Var}(T\mid\Theta)]\\
    &\quad+\operatorname{Var}_\Theta(\mathbb E[T\mid\Theta]).
  \end{aligned}
  \label{eq:phase-total-variance}
\end{equation}
The second term measures variation of the mean delay across
release phases. Under uniform phase mixing, its fraction of total
delay variance is
$R_{\mathrm{between}}=
\operatorname{Var}_\Theta(\mathbb E[T\mid\Theta])/\operatorname{Var}(T)$.
Table~\ref{tab:phase-structure-summary} reports its estimates using
the finite-horizon moment convention in
\appsubref{app:numerical-implementation}{appsub:phase-additional}.

A complementary measure compares the conditional CDFs directly:
\begin{equation}
  \widehat\Delta_{\mathrm{phase},48}
  =\max_{0\leq k,\ell<48}
    \mathcal D_{t_{\mathrm{cmp}}}
    (\widehat F_{\eta_k},\widehat F_{\eta_\ell}).
  \label{eq:phase-sensitivity-ks}
\end{equation}
This maximum is taken over pairs of sampled phases and, through
$\mathcal D_{t_{\mathrm{cmp}}}$, over the comparison time grid.
It responds to distribution-shape changes as well as changes in the
mean; both diagnostics characterize the Volterra reference family.

In Table~\ref{tab:phase-structure-summary}, differences in the
phase-conditioned mean delay account for $3.27\%$ of the mixture
variance in the regular family and $17.7\%$ in the backflow family.
The corresponding maximum CDF separations are $0.3069$ and $0.7690$.
Thus, even the regular family's small between-phase variance fraction
coexists with a substantial difference in cumulative arrivals.
Variation of the mean alone does not capture the full phase dependence,
and uniform phase averaging can conceal distinct conditional profiles.

\begin{table}[!tp]
\caption{Phase Sensitivity of the Two Uniformly Weighted
48-Phase Channel Families}
\label{tab:phase-structure-summary}
\centering
\renewcommand{\arraystretch}{1.15}
\setlength{\tabcolsep}{4pt}
\begin{tabularx}{\columnwidth}{@{}>{\RaggedRight\arraybackslash}Xcc@{}}
\toprule
Diagnostic & Regular & Backflow\\
\midrule
Between-phase variance fraction, $R_{\mathrm{between}}$
  & $3.27\%$ & $17.7\%$\\
Maximum pairwise CDF separation,
$\widehat\Delta_{\mathrm{phase},48}$
  & $0.3069$ & $0.7690$\\
\bottomrule
\end{tabularx}
\end{table}

\subsection{Timing likelihoods and absorption probabilities}
\label{subsec:phase-aware-communication}

In the transition law~\eqref{eq:communication-kernel},
a candidate release time selects both $y-s$ and $\eta(s)$.
Using a single phase-averaged density therefore changes the
likelihood for distinguishing candidate release times.

For example, consider a single-molecule timing alphabet
$\{s_1,\ldots,s_J\}$ with prior probabilities
$\pi_1,\ldots,\pi_J$.
Under the ideal observation model of Section~\ref{sec:phase-resolved-channel},
a maximum-a-posteriori decision is
\begin{equation}
  \widehat j_{\mathrm{MAP}}(y)
  \in
  \arg\max_{1\leq j\leq J}
    \pi_j f_{\eta(s_j)}(y-s_j),
  \label{eq:phase-aware-map-rule}
\end{equation}
where $f_\eta(t)=0$ for $t\leq0$.
Using C-IG replaces $f_{\eta(s_j)}$ by
$g_{\eta(s_j)}$.

For a reception deadline $t_{\mathrm w}$ measured
from release, the absorption probability is
\begin{equation}
  P_{\mathrm{abs}}(s,t_{\mathrm w})
  =
  \Prob(Y\leq s+t_{\mathrm w}\mid S=s)
  =
  F_{\eta(s)}(t_{\mathrm w}).
  \label{eq:phase-conditioned-capture}
\end{equation}
The C-IG estimate is $G_{\eta(s)}(t_{\mathrm w})$.
For $0\leq t_{\mathrm w}\leq t_{\mathrm{cmp}}$,
\begin{equation}
  \begin{aligned}
    &|F_{\eta(s)}(t_{\mathrm w})
      -G_{\eta(s)}(t_{\mathrm w})|\\
    &\qquad\leq
    \mathcal D_{t_{\mathrm{cmp}}}
      (F_{\eta(s)},G_{\eta(s)}).
  \end{aligned}
  \label{eq:capture-probability-error}
\end{equation}
Thus, the CDF metric bounds the absorption-probability error
uniformly over deadlines in the comparison window.
The errors in Section~\ref{sec:cig-accuracy} estimate this
quantity for the tested C-IG settings and numerical grids.
For the phase families in Table~\ref{tab:phase-structure-summary},
$\widehat\Delta_{\mathrm{phase},48}$ instead compares two release
phases at the same elapsed-time deadline. The largest differences in absorption probability over the sampled phases and deadlines are $30.69$ percentage
points for the regular family and $76.90$ percentage points for the
backflow family. These are absolute phase-to-phase differences,
not C-IG approximation errors.

For a release protocol with phase measure $q$, the phase-averaged
absorption probability is $F_q(t_{\mathrm w})$
from~\eqref{eq:general-phase-mixture-cdf}.
Absolute arrival times additionally retain the release-time shift
as described in
\secref{sec:phase-resolved-channel}{subsec:phase-mixture}.
Deadline calculations use CDFs, whereas timing likelihoods
use density values. Preserving cumulative arrivals alone therefore
does not establish accurate timing likelihoods, motivating the joint
assessment of CDF errors, density errors, and peak structure in
Section~\ref{sec:cig-accuracy}.

\subsection{Modeling scope}
\label{subsec:scope-limitations-extensions}

The model isolates deterministic pulsatility in a one-dimensional
link with known parameters, constant diffusivity, and perfect absorption.
Spatial velocity profiles~\cite{Wicke2018DuctFlow}, partial absorption
or reaction~\cite{Ahmadzadeh2016Reactive}, uncertain environmental phase,
and molecule--flow feedback require additional modeling. Repeated
emissions also require molecule counts, receiver observations, and an
intersymbol-interference model. The present results establish channel-law
properties and approximation accuracy; they do not establish capacity
or receiver error-rate improvements. Other smooth deterministic
waveforms can use the cumulative-displacement formulation, subject to
the regularity and long-time assumptions of
Section~\ref{sec:volterra-characterization} and a new C-IG accuracy
assessment.

%% file: sections/08_conclusion.tex
\Needspace{4\baselineskip}
\section{Conclusion}
\label{sec:conclusion}

This work develops a phase-resolved molecular timing-channel
formulation with constant diffusivity and pulsatile drift,
motivated by molecular transport under stable medium conditions
and cardiac-driven flow variation.
The integrated conference C-IG construction provides a tractable
analytical approximation to the resulting first-hitting laws,
including regimes with transient backflow.
The exact Volterra characterization supplies the reference law,
with regularity, uniqueness, properness, and moment properties.
Numerical comparisons show that close cumulative agreement can
coexist with density and peak-count discrepancies.
Phase-family analysis further shows that phase averaging can
conceal differences among conditional arrival profiles.
These results clarify which arrival-time structures C-IG
preserves in the tested regimes and how release-phase information
enters timing likelihoods and absorption probabilities within
a prescribed reception window.

%% file: appendices/05_numerical_implementation.tex
\section{Numerical Implementation Details}
\label{app:numerical-implementation}

\subsection{Refinement and deterministic residuals}
\label{appsub:deterministic-details}

The pulsatile-drift refinement test uses
$(d,D,v_0,a,\omega)=(1,0.2,0.8,0.3,2)$ at $\eta=0.4$,
$t_{\max}=10$, and
$\Delta t\in\{0.016,0.008,0.004\}$.
The finest solution is mapped to the coarser grids by
piecewise cubic Hermite interpolation to evaluate
\begin{equation}
  E_1(\Delta t)
  =
  \int_0^{10}
  \left|
    \widehat f_\eta^{(\Delta t)}(t)
    -
    \widehat f_\eta^{(0.004)}(t)
  \right|dt,
  \label{eq:grid-refinement-error}
\end{equation}
using trapezoidal quadrature.

The independent first-contact check evaluates
\begin{equation}
  \begin{aligned}
    r_{\mathrm F,\eta}(t)
    &=
    \int_0^t
      \overline\Phi\!\left(
        \frac{b_\eta(t)-b_\eta(u)}{\sqrt{t-u}}
      \right)\widehat f_\eta(u)\,du\\
    &\quad-
    \overline\Phi\!\left(
      \frac{b_\eta(t)}{\sqrt t}
    \right),
    \qquad t>0.
  \end{aligned}
  \label{eq:fortet-validation-residual}
\end{equation}
with the tail kernel set to its diagonal limit $1/2$.
The finite-horizon stopped-moment residuals are
\begin{equation}
  \widehat r_{1,\eta}
  =
  \int_0^{t_{\max}}
    [d-M_\eta(t)]\widehat f_\eta(t)\,dt
  \label{eq:first-martingale-diagnostic}
\end{equation}

and
\begin{equation}
  \widehat r_{2,\eta}
  =
  \int_0^{t_{\max}}
    \left([d-M_\eta(t)]^2-2Dt\right)
    \widehat f_\eta(t)\,dt.
  \label{eq:second-martingale-diagnostic}
\end{equation}
with units of length and squared length, respectively.
The upper block of Table~\ref{tab:reference-diagnostics} uses
$\Delta t=10^{-3}$, horizon $15$ for the constant-drift and regular
cases, and horizon $20$ for backflow. Constant-drift parameters are
$(1,0.2,0.8,0,2)$; the regular and backflow parameters are those
in Table~\ref{tab:consolidated-numerical-settings}.
The regular phases are $\eta=0.4+j\pi/2$, $j=0,\ldots,3$,
and the backflow phase is $3\pi/2$.

\subsection{Trajectory simulation and sampling uncertainty}
\label{appsub:trajectory-details}

For Fig.~\ref{fig:volterra-oracle-validation}(c), the regular-pulsation case at $\eta=0.4+\pi/2$ uses $N_{\mathrm{tr}}=1000$,
$\Delta t_{\mathrm{tr}}=0.005$, and horizon $15$.
Deterministic displacement is evaluated exactly at path-grid
endpoints. Missed within-step crossings are treated with a
locally linearized Brownian-bridge construction using eight
bisection levels, following~\cite{Gobet2001}.
The MATLAB generator is \texttt{twister}, with seed $20260806$.
The KS statistic uses both ECDF limits at each ordered hitting time,
with the terminal-normalized reference CDF evaluated by piecewise
cubic Hermite interpolation.

The conference calculation in
Fig.~\ref{fig:cig-direct-comparison} uses
$(d,D,v_0,a,\omega,\eta)=(5,1,1,2,2\pi,0)$ and $10^5$
trajectories, advanced by Euler--Maruyama to time $8$ with
step $10^{-3}$ and left-endpoint drift.
The NumPy generator is \texttt{MT19937}, with seed $42$.
Detected crossings remove the particle, with hitting time
interpolated linearly between step endpoints and no
Brownian-bridge correction. Histogram bins have width $0.05$;
bin counts are divided by $0.05\times10^5$ and cumulative counts
by $10^5$, retaining nonarrivals in both denominators.

For an ECDF based on $N_{\mathrm{tr}}$ independent samples,
the two-sided DKW confidence-band half-width at error probability
$\delta$ is
\cite{Massart1990}
\begin{equation}
  \epsilon_{\mathrm{DKW}}
  =
  \sqrt{
    \frac{\log(2/\delta)}{2N_{\mathrm{tr}}}
  }.
  \label{eq:dkw-radius}
\end{equation}

The value $0.04295$ in
\secref{sec:numerical-methods}{subsec:trajectory-validation}
uses $\delta=0.05$ and $N_{\mathrm{tr}}=1000$.

\subsection{Phase-family grids and moments}
\label{appsub:phase-gates}

All 48 phases within a family share the time step and
horizon in Table~\ref{tab:consolidated-numerical-settings}.
Acceptance requires finite output,
$\epsilon_{\mathrm{mass}}\leq2\times10^{-4}$,
$\epsilon_-\leq10^{-8}$, and
$\min_n\widehat f_{\eta,n}\geq-10^{-8}$.
The positive terminal deficit
\begin{equation}
  \widehat\epsilon_{\mathrm{tail},\eta}
  =
  \max\{0,1-\widehat F_\eta(t_{\max})\}
  \label{eq:numerical-tail-deficit}
\end{equation}
must be at most $10^{-4}$.
If the positive-deficit threshold is exceeded, the phase-family
workflow permits multiplying the common horizon by $1.5$, at most twice.
All 96 laws pass at the horizons in
Table~\ref{tab:consolidated-numerical-settings}, so no expansion is used.

\label{appsub:phase-additional}
For the variance fractions in
Table~\ref{tab:phase-structure-summary}, mass-normalized finite-horizon
moments of each phase law give
\begin{equation}
  \begin{aligned}
    m_{r,k}
      &=\frac{\operatorname{Trap}_{[0,t_{\max}]}
          [t^r\widehat f_{\eta_k}(t)]}
        {\widehat F_{\eta_k}(t_{\max})},\qquad r=1,2,\\
    R_{\mathrm{between}}
      &=\frac{\langle m_{1,k}^2\rangle-\langle m_{1,k}\rangle^2}
        {\langle m_{2,k}\rangle-\langle m_{1,k}\rangle^2},
    \qquad \langle z_k\rangle=\frac1{48}\sum_{k=0}^{47}z_k.
  \end{aligned}
  \label{eq:finite-horizon-variance-fraction}
\end{equation}
These yield $R_{\mathrm{between}}=0.0326848$ and $0.176954$
for the regular and backflow families. Normalization applies only to
these truncated-law moments; the reference CDFs in
\eqref{eq:forty-eight-phase-average-cdf} and
\eqref{eq:phase-sensitivity-ks} are unchanged. No tail contribution is
extrapolated.